\documentclass[a4paper,11pt,twoside]{article}
\usepackage{textcomp}
\usepackage{enumerate}
\usepackage{amsmath}
\usepackage{amssymb}
\usepackage{amsfonts}
\usepackage{amsthm}
\usepackage{mathrsfs}
\usepackage{epsfig}
\usepackage{graphicx}
\DeclareMathAlphabet{\mathpzc}{OT1}{pzc}{m}{it}
\theoremstyle{plain}
\newtheorem{theorem}{Theorem}
\newtheorem{corollary}[theorem]{Corollary}
\newtheorem{lemma}[theorem]{Lemma}
\newtheorem{proposition}[theorem]{Proposition}

\theoremstyle{definition}
\newtheorem{definition}[theorem]{Definition}
\newtheorem{assumption}[theorem]{Assumption}
\newtheorem{condition}[theorem]{Condition}
\theoremstyle{remark}
\newtheorem{remark}[theorem]{Remark}

\title{Risk Sensitive Investment Management with Affine Processes: a Viscosity Approach}
\author{Mark Davis~\footnote{Department of Mathematics, Imperial College London, London SW7 2AZ, England, Email: mark.davis@imperial.ac.uk}~\footnote{The authors are very grateful to the editors and an anonymous referees for a number of very helpful comments.} \;  and S\'ebastien Lleo~\footnote{Department of Mathematics, Imperial College London, London SW7 2AZ, England, Email: sebastien.lleo@imperial.ac.uk}}
\date{\today}
\begin{document}
\bibliographystyle{plain}
\maketitle
%\tableofcontents
%%

\begin{abstract}
	In this paper, we extend the jump-diffusion model proposed by Davis and Lleo to include jumps in asset prices as well as valuation factors. The criterion, following earlier work by Bielecki, Pliska, Nagai and others, is risk-sensitive optimization (equivalent to maximizing the expected growth rate subject to a constraint on variance.) In this setting, the Hamilton-Jacobi-Bellman equation is a partial integro-differential PDE. The main result of the paper is to show that the value function of the control problem is the unique viscosity solution of the Hamilton-Jacobi-Bellman equation. 
\end{abstract}


%%%
\section{Introduction}

In this paper, we extend the jump diffusion risk-sensitive asset management model proposed by Davis and Lleo~\cite{dall_JDRSAM_Diff} to allow jumps in both asset prices and factor levels. 
\\

Risk-sensitive control generalizes classical stochastic control by parametrizing explicitly the  degree of risk aversion or risk tolerance of the optimizing agent. In risk-sensitive control, the decision maker's objective is to select a control policy $h(t)$ to maximize the criterion
\begin{equation}\label{eq_criterion_J}
    J(t,x,h;\theta) := -\frac{1}{\theta}\ln\mathbf{E}\left[e^{-\theta F(t,x,h)} \right]
\end{equation}
where $t$ is the time, $x$ is the state variable, $F$ is a given reward function, and the risk sensitivity $\theta \in ]-1,0[\cup]0,\infty)$ is an exogenous parameter representing the decision maker's degree of risk aversion. A Taylor expansion of this criterion around $\theta = 0$ yields
\begin{equation}\label{eq_Taylor_RSC}
    J(t,x,h;\theta)
    = \mathbf{E}\left[F(t,x,h)\right]
    - \frac{\theta}{2}\mathbf{Var}\left[F(t,x,h)\right]
    + O(\theta^2)
\end{equation}
which shows that the risk-sensitive criterion amounts to maximizing $\mathbf{E}\left[F(t,x,h)\right]$ subject to a penalty for variance. Jacobson~\cite{ja73}, Whittle~\cite{wh90}, Bensoussan and Van Schuppen~\cite{bevs85} led the theoretical development of risk sensitive control while Lefebvre and Montulet~\cite{lemo94}, Fleming~\cite{fl95} and Bielecki and Pliska~\cite{bipl99} pioneered the financial application of risk-sensitive control. In particular, Bielecki and Pliska proposed the logarithm of the investor's wealth as a reward function, so that the investor's objective is to maximize the risk-sensitive (log) return of his/her portfolio or alternatively to maximize a function of the power utility (HARA) of terminal wealth. Bielecki and Pliska brought an enormous contribution to the field by studying the economic properties of the risk-sensitive asset management criterion (see~\cite{bipl03}), extending the asset management model into an intertemporal CAPM (\cite{bipl04}), working on transaction costs (\cite{bipl00}), numerical methods (\cite{bihhpl02}) and considering factors driven by a CIR model (\cite{biplsh05}). Other main contributors include Kuroda and Nagai~\cite{kuna02} who introduced an elegant solution method based on a change of measure argument. Davis and Lleo applied this change of measure technique to solve a benchmarked investment problem in which an investor selects an asset allocation to outperform a given financial benchmark (see~\cite{dall_RSBench}) and analyzed the link between optimal portfolios and fractional Kelly strategies (see~\cite{dall_ZiembaBook}). More recently, Davis and Lleo~\cite{dall_JDRSAM_Diff} extended the risk-sensitive asset management model by allowing jumps in asset prices.
\\

In this chapter, our contribution is to allow not only jumps in asset prices but also in the level of the underlying valuation factors. Once we introduce jumps in the factors, the Bellman equation becomes a nonlinear Partial Integro-Differential equation and an analytical or classical $C^{1,2}$ solutions may not exist. As a result, to give a sense to the relation between the value function and the risk sensitive Hamilton-Jacobi-Bellman Partial Integro Differential Equation (RS HJB PIDE), we consider a class of weak solutions called viscosity solutions, which have gained a widespread acceptance in control theory in recent years. The main results are a comparison theorem and the proof that the value function of the control problem under consideration is the unique continuous viscosity solution of the associated RS HJB PIDE. In particular, the proof of the comparison results uses non-standard arguments to circumvent difficulties linked to the highly nonlinear nature of the RS HJB PIDE and to the unboundedness of the instantaneous reward function $g$.
\\

This chapter is organized as follows. Section 2 introduces the general setting of the model and defines the class of random Poisson measures which will be used to model the jump component of the asset and factor dynamics. In Section 3 we formulate the control problem and apply a change of measure to obtain a simpler auxiliary criterion. Section 4 outlines the properties of the value function.
In Section 5 we show that the value function is a viscosity solution of the RS HJB PIDE before proving a comparison result in Section 6 which provides uniqueness.
\\

%%%%%%%%%%%%%%%%%%%%%%%%%%%%%%%%%%%%%%%%%%%%%%%%%%%%%%%%%%%%%%%%%%%%%%%%%%%%%%%%%%%
%
%   NEXT SECTION
%
%%%%%%%%%%%%%%%%%%%%%%%%%%%%%%%%%%%%%%%%%%%%%%%%%%%%%%%%%%%%%%%%%%%%%%%%%%%%%%%%%%%

%%%
\section{Analytical Setting}
Our analytical setting is based on that of~\cite{dall_JDRSAM_Diff}. The notable difference is that we allow the factor processes to experience jumps.
\\
 
\subsection{Overview}
The growth rates of the assets are assumed to depend on $n$ valuation factors $X_1(t), \ldots, X_n(t)$ which follow the dynamics given in equation~\eqref{eq_FactorProcess} below. The assets market comprises $m$ risky securities $S_i, \; i=1,\ldots m$. Let $M := n+m$. Let $(\Omega, \left\{ \mathcal{F}_{t} \right\}, \mathcal{F},
\mathbb{P})$ be the underlying probability space. On this space is
defined an $\mathbb{R}^M$-valued
$\left(\mathcal{F}_t\right)$-Brownian motion $W(t)$ with components
$W_k(t)$, $k=1,\ldots,M$. Moreover, let $(\mathbf{Z},\mathcal{B}_{\mathbf{Z}})$ be a Borel space\footnote{$\mathbf{Z}$ is a standard measurable (metric or topological) space and $\mathcal{B}_{\mathbf{Z}}$ is the Borel $\sigma$-field endowed to $\mathbf{Z}$.}. Let $\textbf{p}$ be an $(\mathcal{F}_t)$-adapted $\sigma$-finite Poisson point process on $\mathbf{Z}$ whose underlying point functions are maps from a countable set $\mathbf{D}_{\textbf{p}} \subset (0,\infty)$ into $\mathbf{Z}$. Define
\begin{equation}\label{def_JDRSAM_ZFrak_set}
    \mathfrak{Z}_\textbf{p} := \left\{ U \in \mathcal{B}(Z), \mathbb{E} \left[N_{\textbf{p}}(t,U)\right] < \infty \; \forall t\right\}
\end{equation}
Consider $N_{\textbf{p}}(dt,dz)$, the Poisson random measure on
$(0,\infty) \times \mathbf{Z}$ induced by $\textbf{p}$. Following Davis and Lleo~\cite{dall_JDRSAM_Diff}, we concentrate on stationary Poisson point processes of class (QL) with associated Poisson random measure
$N_{\textbf{p}}(dt,dx)$. The class (QL) is defined in~\cite{ikwa81}
(Definition II.3.1 p. 59) as
\begin{definition}
    An $(\mathcal{F}_t)$-adapted point process $\textbf{p}$ on $(\Omega,\mathcal{F},\mathbb{P})$ is said to be \emph{of class (QL)} with respect to $(\mathcal{F}_t)$ if it is $\sigma$-finite and there exists $\hat{N}_{\textbf{p}} = \left(\hat{N}_{\textbf{p}}(t,U)\right)$ such that
\begin{enumerate}[(i.)]
\item for $U \in \mathfrak{Z}_p$, $t \mapsto \hat{N}_{\textbf{p}}(t,U)$ is a continuous $(\mathcal{F}_t)$-adapted increasing process;
\item for each $t$ and a.a. $\omega \in \Omega$, $U \mapsto \hat{N}_{\textbf{p}}(t,U)$ is a $\sigma$-finite measure on $(\mathbf{Z},\mathcal{B}(\mathbf{Z}))$;
\item for $U \in \mathfrak{Z}_p$, $t \mapsto \tilde{N}_{\textbf{p}}(t,U) = N_{\textbf{p}}(t,U) - \hat{N}_{\textbf{p}}(t,U)$ is an $(\mathcal{F}_t)$-martingale;
\end{enumerate}
The random measure $\left\{\hat{N}_{\textbf{p}}(t,U)\right\}$ is
called the \emph{compensator} of the point process $p$.
\end{definition}
Since the Poisson point processes we consider are stationary, then
their compensators are of the form $\hat{N}_{\textbf{p}}(t,U) =
\nu(U)t$ where $\nu$ is the $\sigma$-finite characteristic measure
of the Poisson point process $\textbf{p}$. For notational convenience, we define the Poisson random
measure $\bar{N}_{\textbf{p}}(dt,dz)$ as
\begin{eqnarray}
    &&\bar{N}_{\textbf{p}}(dt,dz)
                                                \nonumber\\
    &=& \left\{ \begin{array}{ll}
        N_{\textbf{p}}(dt,dz) - \hat{N}_{\textbf{p}}(dt,dz) = N_{\textbf{p}}(dt,dz) - \nu(dz)dt =: \tilde{N}_{\textbf{p}}(dt,dz)    &   \textrm{if } z \in \mathbf{Z}_0     \\
        N_{\textbf{p}}(dt,dz)                      &   \textrm{if } z \in \mathbf{Z} \backslash \mathbf{Z}_0       \\
    \end{array}\right.
                                                \nonumber
\end{eqnarray}
where $\mathbf{Z}_0 \subset \mathcal{B}_{\mathbf{Z}}$ such that
$\nu(\mathbf{Z} \backslash \mathbf{Z}_0)<\infty$.
\\

\subsection{Factor Dynamics}\label{Chapter3_JDRSAM_theory_factordynamics}
We model the dynamics of the $n$ factors with an affine jump diffusion process
\begin{equation}\label{eq_FactorProcess}
    dX(t) = (b + BX(t^-))dt + \Lambda dW(t) + \int_{\mathbf{Z}}\xi(z)\bar{N}_{\textbf{p}}(dt,dz),
    \qquad X(0) = x
\end{equation}
where $X(t)$ is the $\mathbb{R}^{n}$-valued factor process with components $X_{j}(t)$ and $b \in \mathbb{R}^n$, $B \in \mathbb{R}^{n\times n}$, $\Lambda := \left[ \Lambda_{ij} \right], \; i = 1, \ldots, n, \; j =1, \ldots, N$ and $\xi(z) \in \mathbb{R}^n$ with $ -\infty < \xi_{i}^{min} \leq \xi_{i}(z) \leq \xi_{i}^{max} < \infty$ for $i = 1, \ldots, n$. Moreover, the vector-valued function $\xi(z)$ satisfies:
\begin{equation}
    \int_{\mathbf{Z}_0} \lvert\xi(z)\rvert^2 \nu(dz) < \infty
                                            \nonumber
\end{equation}
(see for example Definition II.4.1 in Ikeda and Watanabe~\cite{ikwa81} where $\mathbf{F}_{\textbf{P}}$ and $\mathbf{F}_{\textbf{P}}^{2,loc}$ are respectively given in equations II(3.2) and II(3.5))\\

\subsection{Asset Market Dynamics}\label{Chapter3_JDRSAM_theory_assetdynamics}
Let $S_0$ denote the wealth invested in the money market account with dynamics given by the equation:
\begin{equation}\label{eq_JDRSAM_BankAccount}
    \frac{dS_{0}(t)}{S_{0}(t)} = \left(a_0 + A_0'X(t)\right)dt,
    \qquad S_0(0) = s_0
\end{equation}
where $a_0 \in \mathbb{R}$ is a scalar constant, $A_0 \in
\mathbb{R}^{n}$ is a $n$-element column vector and where $M$' denotes the transposed matrix of $M$. Note that if we set
$A_0 = 0$ and $a_0 = r$, then
equation~\eqref{eq_JDRSAM_BankAccount} can be interpreted as the
dynamics of a globally risk-free asset. Let $S_{i}(t)$ denote the price at time $t$ of the $i$th security, with $i = 1,
\ldots, m$. The dynamics of risky security $i$ can be expressed as:
\begin{eqnarray}\label{eq_SecurityProcess}
    \frac{dS_{i}(t)}{S_{i}(t^{-})} &=&
        (a + AX(t))_{i}dt
        + \sum_{k=1}^{N} \sigma_{ik} dW_{k}(t)
        + \int_{\mathbf{Z}} \gamma_i(z)\bar{N}_{\textbf{p}}(dt,dz),
                                                \nonumber\\
    && S_i(0) = s_i,
    \quad i = 1, \ldots, m
\end{eqnarray}
where $a \in \mathbb{R}^m$, $A \in \mathbb{R}^{m \times n }$,
$\Sigma := \left[ \sigma_{ij} \right], \; i = 1, \ldots, m, \; j =
1, \ldots, M$ and $\gamma(z) \in \mathbb{R}^m$ satisfies Assumption~\ref{as_assetjumps_upanddown_1}
\\

\begin{assumption}\label{as_assetjumps_upanddown_1}
$\gamma(z) \in \mathbb{R}^m$ satisfies
\begin{eqnarray}
	-1 \leq \gamma_{i}^{min} \leq \gamma_{i}(z) \leq \gamma_{i}^{max} < +\infty
	, \qquad i =1,\ldots,m
											\nonumber
\end{eqnarray}and
\begin{eqnarray}
	-1 \leq \gamma_{i}^{min} < 0 < \gamma_{i}^{max} < +\infty
	, \qquad i =1,\ldots,m
											\nonumber
\end{eqnarray}
for $i = 1, \ldots, m$. Furthermore, define
\begin{equation}
    \mathbf{S} := \textrm{supp}(\nu) \in \mathcal{B}_{\textbf{Z}}
                                                \nonumber
\end{equation}
and
\begin{equation}
    \tilde{\mathbf{S}}
    :=  \textrm{supp}(\nu \circ \gamma^{-1})
    \in \mathcal{B}\left(\mathbb{R}^m\right)
                                                \nonumber
\end{equation}
where $\textrm{supp}(\cdot)$ denotes the measure's support, then we assume that $\prod_{i=1}^{m} [\gamma_{i}^{min}, \gamma_{i}^{max}]$ is the smallest closed hypercube containing $\tilde{\mathbf{S}}$.
\\

In addition, the vector-valued function $\gamma(z)$ satisfies:
\begin{equation}\label{as_assetjumps_gamma_integrable}
    \int_{\mathbf{Z}_0} \lvert\gamma(z)\rvert^2 \nu(dz) < \infty
                                            \nonumber
\end{equation}
\end{assumption}

As note in~\cite{dall_JDRSAM_Diff}, Assumption~\ref{as_assetjumps_upanddown_1} requires that each asset has, with positive probability, both upward and downward jump and as a result bounds the space of controls. 
\\

Define the set $\mathcal{J}$ as
\begin{equation}\label{def_JDRSAM_setJ}
    \mathcal{J} := \left\{h \in \mathbb{R}^m:  -1-h'\psi < 0 \quad \forall \psi \in \tilde{\mathbf{S}}\right\}
\end{equation}
For a given $z$, the equation $h'\gamma(z)   = -1$ describes a
hyperplane in $\mathbb{R}^m$. Under Assumption~\ref{as_assetjumps_upanddown_1} $\mathcal{J}$ is a convex subset of $\mathbb{R}^m$.
\\

\subsection{Portfolio Dynamics}\label{Chapter3_JDRSAM_theory_portfoliodynamics}
We will assume that:
\begin{assumption}\label{as_JDRSAM_sigmaposdef}
    The matrix $\Sigma\Sigma'$ is positive definite.
\end{assumption}
and
\begin{assumption}\label{as_JDRSAM_uncorrelatedjumps}
    The systematic (factor-driven) and idiosyncratic (asset-driven) jump risks are uncorrelated, i.e $\forall z \in \mathbf{Z}$ and $i=1,\ldots,m$, $\gamma_i(z)\xi'(z) = 0$.
\end{assumption}
The second assumption implies that there cannot be simultaneous jumps in the factor process and any asset price process. This assumption, which will prove sufficient to show the existence of a unique optimal investment policy, may appear somewhat restrictive as it does not enable us to model a jump correlation structure across factors and assets, although we can model a jump correlation structure within the factors and within the assets.
\\

\begin{remark}
Assumption~\eqref{as_JDRSAM_uncorrelatedjumps} is automatically satisfied when jumps are only allowed in the security prices and the state variable $X(t)$ is modelled using a diffusion process (see~\cite{dall_JDRSAM_Diff} for a full treatment of this case).
\\
\end{remark}

Let $\mathcal{G}_t := \sigma((S(s), X(s)), 0 \leq s \leq t)$ be the
sigma-field generated by the security and factor processes up to
time $t$.
\\

An \textit{investment strategy} or \textit{control process} is an $\mathbb{R}^m$-valued process with the interpretation that $h_i(t)$ is the fraction of current portfolio value invested in the $i$th asset, $i=1,\ldots,m$. The fraction invested in the money market account is then $h_0(t) = 1 - \sum_{i=1}^{m} h_{i}(t)$.
\\

\begin{definition}\label{def_JDRSAM_controlprocess_h}
    An $\mathbb{R}^m$-valued control process $h(t)$ is in class $\mathcal{H}$ if the
    following conditions are satisfied:
    \begin{enumerate}
        \item $h(t)$ is progressively measurable with respect to
        $\left\{ \mathcal{B}([0,t]) \otimes \mathcal{G}_t\right\}_{t \geq
        0}$ and is c\`adl\`ag;

        \item $P\left(\int_{0}^{T} \left| h(s) \right|^2 ds < +\infty \right)
        =1, \quad \forall T>0$;

        \item $h'(t)\gamma(z) > -1, \quad \forall t >0, z \in \mathbf{Z}$, a.s. $d\nu$.
    \end{enumerate}
\end{definition}

Define the set $\mathcal{K}$ as
\begin{equation}\label{def_JDRSAM_setmathcalK}
    \mathcal{K} := \left\{h(t) \in \mathcal{H}: h(t) \in \mathcal{J} \quad \forall t \textrm{ a.s.}\right\}
\end{equation}

\begin{lemma}
	Under Assumption~\ref{as_assetjumps_upanddown_1}, a control process $h(t)$ satisfying condition 3 in Definition~\ref{def_JDRSAM_controlprocess_h} is bounded.
\end{lemma}

\begin{proof}
	The proof of this result is immediate.
\end{proof}

\begin{definition}\label{def_JDRSAM_admissible_A}
    A control process $h(t)$ is in class $\mathcal{A}(T)$ if the
    following conditions are satisfied:
    \begin{enumerate}
        \item $h(t) \in \mathcal{H}$ $\forall t \in [0,T]$;

			\item $\mathbf{E} \chi_T^h= 1$ where $\chi_t^h$ is the Dol\'eans exponential defined as
\begin{eqnarray}\label{eq_JDRSAM_Doleansexp_chi}
    \chi_t^h
    &:=& \exp \left\{ -\theta \int_{0}^{t} h(s)'\Sigma dW_s
    -\frac{1}{2} \theta^2 \int_{0}^{t} h(s)'\Sigma\Sigma'h(s) ds            \right.
                                                            \nonumber\\
   &&   \left.
        +\int_{0}^{t} \int_{\mathbf{Z}} \ln\left(1-G(z,h(s);\theta)\right) \tilde{N}_{\textbf{p}}(ds,dz)
            \right.
                                                            \nonumber\\
   &&   \left.
        +\int_{0}^{t} \int_{\mathbf{Z}} \left\{\ln\left(1-G(z,h(s);\theta)\right)+G(z,h(s);\theta)\right\}\nu(dz)ds
    \right\},
																				\nonumber\\
\end{eqnarray}
and
\begin{eqnarray}
    G(z,h;\theta) &=&
        1-\left(1+h'\gamma(z)\right)^{-\theta}
\end{eqnarray}

\end{enumerate}
\end{definition}

\begin{definition}
    We  say that a control process $h(t)$ is \emph{admissible} if $h(t) \in \mathcal{A}(T)$.
\end{definition}

The proportion invested in the money market account is $h_0(t)=1-\sum_{i=1}^{m} h_i(t)$. Taking this budget equation into consideration, the wealth $V(t,x,h)$, or $V(t)$, of the investor in response to an investment strategy $h(t) \in \mathcal{H}$, follows the dynamics\begin{eqnarray}
   \frac{dV(t)}{V(t^-)}
    &=& \left(a_0 + A_0'X(t)\right)dt
            + h'(t)\left(a-a_0\mathbf{1}
            +\left(A-\mathbf{1}A_0'\right)X(t)\right)dt
                                                    \nonumber\\
    &&
            + h'(t)\Sigma dW_t
            + \int_{\mathbf{Z}}h'(t)\gamma(z)\bar{N}_{\textbf{p}}(dt,dz)
                                        \nonumber
\end{eqnarray}
where $\mathbf{1} \in \mathbf{R}^m$ denotes the $m$-element unit column vector and with $V(0) = v$. Defining $\hat{a} := a - a_0\mathbf{1}$ and $\hat{A} := A -
\mathbf{1}A_0'$, we can express the portfolio dynamics as
\begin{eqnarray}\label{eq_JDRSAM_V_dynamics}
    \frac{dV(t)}{V(t^-)}
        = \left(a_0 + A_0'X(t)\right)dt
        + h'(t)\left(\hat{a}+\hat{A}X(t)\right)dt
        + h'(t)\Sigma dW_t
        + \int_{\mathbf{Z}}h'(t)\gamma(z)\bar{N}_{\textbf{p}}(dt,dz)
                                        \nonumber\\
\end{eqnarray}

\section{Problem Setup}
\subsection{Optimization Criterion}
We will follow Bielecki and Pliska~\cite{bipl99} and Kuroda and Nagai~\cite{kuna02} and assume that the objective of the investor is to maximize the long-term risk adjusted growth of his/her portfolio of assets. In this context, the objective of the risk-sensitive management problem is to find $h^*(t) \in \mathcal{A}(T)$ that maximizes the control criterion
\begin{equation}\label{eq_JDRSAM_criterion_J}
    J(t,x,h;\theta) := -\frac{1}{\theta}\ln\mathbf{E}\left[e^{-\theta \ln V(t,x,h)} \right]
\end{equation}
\\

By It\^o, the log of the portfolio value in response to a strategy $h$ is
\begin{eqnarray}\label{eq_JDRSAM_Vt}
    \ln V(t)
        &=&\ln v + \int_{0}^{t} \left(a_0 + A_0'X(s)\right)
        +   h(s)'\left(\hat{a}+\hat{A}X(s)\right)ds
        -   \frac{1}{2} \int_{0}^{t} h(s)'\Sigma\Sigma'h(s) ds
                                                                      \nonumber \\
        &&  +\int_{0}^{t} h(s)'\Sigma dW(s)
                                                                            \nonumber \\
        &&  +\int_{0}^{t}\int_{\mathbf{Z}_0}
            \left\{
            \ln\left(1+h(s)'\gamma(z)\right)-h(s)'\gamma(z)
            \right\}\nu(dz)ds
                                                                            \nonumber \\
        &&  +\int_{0}^{t} \int_{\mathbf{Z}} \ln\left(1+h(s)'\gamma(z)\right) \bar{N}_{\textbf{p}}(ds,dz)
\end{eqnarray}
Hence,
\begin{eqnarray}\label{eq_JDRSAM_eminthetaVt}
    e^{-\theta \ln V(t)} &=&  v^{-\theta}
        \exp \left\{ \theta \int_{0}^{t} g(X_s,h(s);\theta) ds \right\} \chi_t^h
\end{eqnarray}
where
\begin{eqnarray}\label{eq_JDRSAM_g_func_def}
    g(x,h;\theta)
    &=&\frac{1}{2} \left(\theta+1 \right)h'\Sigma\Sigma'h- a_0-A_0'x -h'(\hat{a} + \hat{A}x)
                                                            \nonumber\\
   && +\int_{\mathbf{Z}}  \left\{\frac{1}{\theta}
        \left[\left(1+h'\gamma(z)\right)^{-\theta}-1\right]
        +h'\gamma(z)\mathit{1}_{\mathbf{Z}_0}(z)
        \right\} \nu(dz)
\end{eqnarray}
and the Dol\'eans exponential $\chi_t^h$ is given by~\eqref{eq_JDRSAM_Doleansexp_chi}.
\\

\subsection{Change of Measure}
Let $\mathbb{P}_{h}^{\theta}$ be the measure on
$(\Omega,\mathcal{F})$ defined as
\begin{eqnarray}\label{eq_JDRSAM_RNder_chi}
    \left. \frac{d\mathbb{P}_{h}^{\theta}}{d\mathbb{P}}\right|_{\mathcal{F}_t}
	&:=& 
    \chi_t
\end{eqnarray}
For a change of measure to be possible, we must ensure that the
following technical condition holds:
\begin{equation}
    G(z,h(s);\theta) < 1
                                                        \nonumber
\end{equation}
for all $s \in[0,T]$ and $z$ a.s. $d\nu$. This condition is satisfied iff
\begin{eqnarray}\label{cond_JDRSAM_changeofmeasure}
    h'(s)\gamma(z)    > -1
\end{eqnarray}
a.s. $d\nu$, which was already one of the conditions required for
$h$ to be in class $\mathcal{H}$ (Condition 3 in
Definition~\ref{def_JDRSAM_controlprocess_h}).
\\

$\mathbb{P}_{h}^{\theta}$ is a probability measure for $h \in \mathcal{A}(T)$. For $h \in \mathcal{A}(T)$,
\begin{equation}
    W_{t}^{h} = W_t + \theta \int_{0}^{t} \Sigma'h(s) ds
            \nonumber
\end{equation}
is a standard Brownian motion under the measure $\mathbb{P}_{h}^{\theta}$ and we define the $\mathbb{P}_{h}^{\theta}$ compensated Poisson measure as
\begin{eqnarray}
    \int_{0}^{t}\int_{\mathbf{Z}}\tilde{N}_{\textbf{p}}^{h}(ds,dz)
&=&     \int_{0}^{t}\int_{\mathbf{Z}}N_{\textbf{p}}(ds,dz)
    -   \int_{0}^{t}\int_{\mathbf{Z}} \left\{1-G(z,h(s);\theta)\right\}\nu(dz)ds
                \nonumber\\
&=&     \int_{0}^{t}\int_{\mathbf{Z}}N_{\textbf{p}}(ds,dz)
    -   \int_{0}^{t}\int_{\mathbf{Z}} \left\{\left(1+h'\gamma(z)\right)^{-\theta}\right\} \nu(dz)ds
                \nonumber
\end{eqnarray}
As a result, $X(s), \; 0 \leq s \leq t$ satisfies the SDE:
\begin{eqnarray}\label{eq_JDRSAM_state_SDE}
    dX(s)
    &=&     f\left(X(s^-),h(s);\theta\right)ds
				+ \Lambda dW_{s}^{h}
            + \int_{\mathbf{Z}}\xi(z)\tilde{N}_{\textbf{p}}^{h}(ds,dz)
\end{eqnarray}
where
\begin{align}\label{eq_JDRSAM_func_f}
	f(x,h;\theta) 
	:= 	b + Bx 
		- 	\theta\Lambda \Sigma'h
    	+	\int_{\mathbf{Z}}\xi(z)\left[
              \left(1+h'\gamma(z)\right)^{-\theta}
                - \mathit{1}_{\mathbf{Z}_0}(z)
         \right]\nu(dz)
\end{align}

We will now introduce the following two auxiliary criterion functions under the
measure $\mathbb{P}_{h}^{\theta}$:
\begin{itemize}
\item the auxiliary function directly associated with the risk-sensitive control problem:
\begin{equation}\label{eq_JDRSAM_auxcriterion}
    I(v,x;h;t,T;\theta) = - \frac{1}{\theta} \ln \mathbf{E}_{t,x}^{h,\theta}
        \left[ \exp \left\{ \theta \int_{t}^{T} g(X_s,h(s);\theta) ds
        - \theta \ln{v} \right\} \right]
\end{equation}
where $\mathbf{E}_{t,x}^{h,\theta} \left[ \cdot \right]$ denotes the
expectation taken with respect to the measure $\mathbb{P}_{h}^{\theta}$ and with initial conditions $(t,x)$.

\item the exponentially transformed criterion
\begin{equation}\label{eq_JDRSAM_Exp_of_int_criterion}
    \tilde{I}(v,x,h;t,T;\theta)
        := \mathbf{E}_{t,x}^{h,\theta}
        \left[ \exp \left\{ \theta \int_{t}^{T} g(X_s,h(s);\theta)
        ds -\theta \ln v
        \right\} \right]
\end{equation}
which we will find convenient to use in our derivations.
\\
\end{itemize}

We have completed our reformulation of the problem under the measure $\mathbb{P}_{h}^{\theta}$. The state dynamics~\eqref{eq_JDRSAM_state_SDE} is a jump-diffusion process and our objective is to maximize the criterion~\eqref{eq_JDRSAM_auxcriterion} or alternatively minimize~\eqref{eq_JDRSAM_Exp_of_int_criterion}.

\subsection{The HJB Equation}
In this section we derive the risk-sensitive Hamilton-Jacobi-Bellman partial integro differential equation (RS HJB PIDE) associated with the optimal control problem. Since we do not anticipate that a classical solution generally exists, we will not attempt to derive a verification theorem. Instead, we will show that the value function $\Phi$ is a solution of the RS HJB PIDE in the viscosity sense. In fact, we will show that the value function is the unique continuous viscosity solution of the RS HJB PIDE. This result will in turn justify the association of the RS HJB PIDE with the control problem and replace the verification theorem we would derive if a classical solution existed.
\\

Let $\Phi$ be the value function for the auxiliary criterion
function $I(v,x;h;t,T)$ defined in~\eqref{eq_JDRSAM_auxcriterion}. Then $\Phi$ is defined as
\begin{equation}\label{eq_JDRSAM_valuefunction}
    \Phi(t,x) = \sup_{h \in \mathcal{A}(T)} I(v,x;h;t,T)
\end{equation}
We will show that $\Phi$ satisfies the HJB PDE
\begin{equation}\label{eq_JDRSAM_HJBPDE}
    \frac{\partial \Phi}{\partial t}(t,x)
    + \sup_{h \in \mathcal{J}}
    L_{t}^{h}\Phi(t,X(t)) = 0
\end{equation}
where
\begin{eqnarray}\label{eq_JDRSAM_HJBPDE_operator_L}
    L_{t}^{h}\Phi(t,x)
    &=&	f(x,h;\theta)'D\Phi
        + \frac{1}{2} \textrm{tr} \left( \Lambda \Lambda' D^2 \Phi \right)
        - \frac{\theta}{2} (D\Phi)'\Lambda \Lambda' D\Phi
                                \nonumber\\
    &&
    + \int_{\mathbf{Z}} \left\{
        - \frac{1}{\theta}\left(
                e^{-\theta \left(\Phi(t,x+\xi(z)) - \Phi(t,x) \right)} -1
            \right)
        - \xi'(z) D\Phi
        \right\} \nu(dz)
        - g(x,h;\theta)
\end{eqnarray}
$D\cdot = \frac{\partial \cdot}{\partial x}$, and subject to terminal condition
\begin{equation}\label{eq_JDRSAM_HJBPDE_termcond}
	\Phi(T, x) = \ln v
\end{equation}
\\

Similarly, let $\tilde{\Phi}$ be the value function for the auxiliary criterion
function $\tilde{I}(v,x;h;t,T)$. Then $\tilde{\Phi}$ is defined as
\begin{equation}\label{eq_JDRSAM_exptrans_valuefunction}
    \tilde{\Phi}(t,x) = \inf_{h \in \mathcal{A}(T)} \tilde{I}(v,x;h;t,T)
\end{equation}
The corresponding HJB PDE is
\begin{eqnarray}\label{eq_JDRSAM_exptrans_HJBPDE}
  	&&	\frac{\partial \tilde{\Phi}}{\partial t}(t,x)
    	+ \frac{1}{2} \textrm{tr} \left( \Lambda \Lambda' D^2 \tilde{\Phi}(t,x)\right)
		+ H(x,\tilde{\Phi},D\tilde{\Phi})
												\nonumber\\
	&&	+ \int_{\mathbf{Z}} \left\{
              \tilde{\Phi}(t,x+\xi(z))
            - \tilde{\Phi}(t,x)
            - \xi'(z) D\tilde{\Phi}(t,x)
        \right\} \nu(dz)
    = 0
\end{eqnarray}
subject to terminal condition
\begin{eqnarray}\label{eq_JDRSAM_exptrans_HJBPDE_termcond}
	\tilde{\Phi}(T, x) = v^{-\theta}
\end{eqnarray}
and where
\begin{eqnarray}\label{eq_JDRSAM_exptrans_H_function}
   H(s,x,r,p) &=& \inf_{h \in \mathcal{J}} \left\{
        \left(b+ B x - \theta \Lambda \Sigma'h(s) \right)'p
 			+ \theta g(x,h;\theta) r
         \right\}
\end{eqnarray}
for $r \in \mathbb{R}$, $p \in \mathbb{R}^n$ and in particular,
\begin{eqnarray}\label{eq_JDRSAM_relationship_Phi_tildePhi}
    \tilde{\Phi}(t,x)
    &=& \exp \left\{-\theta \Phi(t,x) \right\}
\end{eqnarray}
\\

The supremum in~\eqref{eq_JDRSAM_HJBPDE} can be expressed as
\begin{eqnarray}\label{eq_JDRSAM_supL_deriv}
        && \sup_{h \in \mathcal{J}} L_{t}^{h}\Phi
                    \nonumber\\
        &=&
            \left( b+ Bx \right)'D\Phi
            + \frac{1}{2} \textrm{tr} \left( \Lambda \Lambda' D^2 \Phi \right)
            - \frac{\theta}{2} (D\Phi)'\Lambda \Lambda' D\Phi
            +a_0+A_0'x
            \nonumber\\
            &&
            +\int_{\mathbf{Z}} \left\{
            -\frac{1}{\theta}\left( e^{-\theta \left(\Phi(t,x+\xi(z)) - \Phi(t,x) \right)} -1 \right)
            -\xi'(z)D\Phi\mathit{1}_{\mathbf{Z}_0}(z)
            \right\} \nu(dz)
                \nonumber\\
        &&
            +\sup_{h \in \mathcal{J}} \left\{
            - \frac{1}{2} \left(\theta+1 \right)h'\Sigma\Sigma'h
            -\theta h'\Sigma\Lambda'D\Phi
            +h'(\hat{a} + \hat{A}x)
                \right. \nonumber\\
            &&\left.
            -\frac{1}{\theta}\int_{\mathbf{Z}}\left\{
            		\left(1 - \theta\xi'(z)D\Phi\right)
						\left[\left(1+h'\gamma(z)\right)^{-\theta}-1\right]
					+\theta h'\gamma(z)\mathit{1}_{\mathbf{Z}_0}(z)
                    \right\}\nu(dz)
            \right\}
\end{eqnarray}
Under Assumption~\ref{as_JDRSAM_sigmaposdef} the term
\begin{eqnarray}
    - \frac{1}{2} \left(\theta+1 \right)h'\Sigma\Sigma'h
    -\theta h'\Sigma\Lambda'D\Phi
    +h'(\hat{a} + \hat{A}x)
    - \int_{\mathbf{Z}} h'\gamma(z)\mathit{1}_{\mathbf{Z}_0}(z)\nu(dz)
                                        \nonumber
\end{eqnarray}
is strictly concave in $h$. Under Assumption~\ref{as_JDRSAM_uncorrelatedjumps}, the nonlinear jump-related term
\begin{eqnarray}
	-\frac{1}{\theta}\int_{\mathbf{Z}}\left\{
   	\left(1 - \theta\xi'(z)D\Phi\right)
	   \left[\left(1+h'\gamma(z)\right)^{-\theta}-1\right]
    \right\}\nu(dz)
																\nonumber
\end{eqnarray}
simplifies to
\begin{eqnarray}
	-\frac{1}{\theta}\int_{\mathbf{Z}}
       \left\{
	       \left[\left(1+h'\gamma(z)\right)^{-\theta}-1\right]
	   \right\}\nu(dz)	
																\nonumber
\end{eqnarray}
which is also concave in $h$ $\forall z \in \mathbf{Z}$ a.s. $d\nu$. Therefore, the supremum is reached for a unique optimal control $h^*$, which is an interior point of the set $\mathcal{J}$ defined in equation~\eqref{def_JDRSAM_setJ}, and the supremum, evaluated at $h^*$, is finite.
\\

%%%
\section{Properties of the Value Function}
\subsection{``Zero Beta'' Policies}
As in~\cite{dall_JDRSAM_Diff}, we will use ``zero beta'' ($0\beta$) policies (initially introduced by Black~\cite{bl72})).
\\

\begin{definition}[$0\beta$-policy]\label{def_JDRSAM_ZeroBetaPolicy}
	By reference to the definition of the function $g$ in equation~\eqref{eq_JDRSAM_g_func_def}, a \emph{`zero beta' ($0\beta$) control policy} $\check{h}(t)$ is an admissible control policy for which the function $g$ is independent from the state variable $x$.
\end{definition}

In our problem, the set $\mathcal{Z}$ of $0\beta$-policies is the set of admissible policies $\check{h}$ which satisfy the equation
\begin{eqnarray}
	\check{h}'\hat{A} = - A_0
									\nonumber
\end{eqnarray}
As $m > n$, there is potentially an infinite number of $0\beta$-policies as long as the following assumption is satisfied
\begin{assumption}\label{as_JDRSAM_A_rank_n}
	The matrix $\hat{A}$ has rank $n$.
\end{assumption}

Without loss of generality, we fix a $0\beta$ control $\check{h}$ as a constant function of time so that
\begin{eqnarray}
	g(x,\check{h};\theta) = \check{g}
															\nonumber
\end{eqnarray}
where $\check{g}$ is a constant.
\\

\subsection{Convexity}
\begin{proposition}\label{prop_JDRSAM_convexity_Phi}
	The value function $\Phi(t,x)$ is convex in $x$.
\end{proposition}

\begin{proof}
	See the proof of Proposition 6.2 in~\cite{dall_JDRSAM_Diff}.
\end{proof}

\begin{corollary}\label{coro_JDRSAM_convexity_equivprop_tildePhi}
	The exponentially transformed value function $\tilde{\Phi}$ has the following property: $\forall (x_1,x_2) \in \mathbb{R}^2, \kappa \in (0,1,)$,
\begin{eqnarray}\label{eq_JDRSAM_tildePhi_property_convexityofPhi}
	\tilde{\Phi}(t, \kappa x_1 + (1-\kappa) x_2)
\geq	
	\tilde{\Phi}^\kappa (t, x_1) \tilde{\Phi}^{1-\kappa} (t, x_2)
\end{eqnarray}		
\end{corollary}

\begin{proof}
The property follows immediately from the definition of $\Phi(t, x) = -\frac{1}{\theta} \ln \tilde{\Phi}(t, x)$.
\end{proof}

\subsection{Boundedness}

\begin{proposition}\label{prop_JDRSAM_tildePhi_bounded}
	The exponentially transformed value function $\tilde{\Phi}$ is positive and bounded, i.e. there exists $M>0$ such that
\begin{eqnarray}
		0 \leq \tilde{\Phi}(t,x)	\leq \check{M}	
			\qquad \forall (t,x) \in [0,T]\times\mathbb{R}^n
														\nonumber
\end{eqnarray}
\end{proposition}

\begin{proof}
By definition, 
\begin{eqnarray}
	\tilde{\Phi}(t,x) 
	&=&	\inf_{h \in \mathcal{A(T)}} \mathbf{E}_{t,x}^{h,\theta}
        		\left[ \exp \left\{ \theta \int_{t}^{T} g(X_s,h(s);\theta)ds 
						-\theta \ln v
        		\right\} \right]
	\geq 0
												\nonumber
\end{eqnarray}
Consider the zero-beta policy $\check{h}$. By the Dynamic Programming Principle 
\begin{eqnarray}
	\tilde{\Phi}(t,x) 
	&\leq& 	e^{\theta \left[\int_{t}^{T} g(X(s),\check{h};\theta)ds 
					- \ln v \right]} 
	=			e^{\theta \left[	\check{g}(T-t) - \ln v \right]}
												\nonumber
\end{eqnarray}
which concludes the proof.

\end{proof}

\subsection{Growth}

\begin{assumption}\label{as_JDRSAM_ODEs_alpha_beta}
There exist $2n$ constant controls $\bar{h}^k, k=1,\ldots,2n$ such that
the $2n$ functions $\beta^k: [0,T] \to \mathbb{R}^n$ defined by
\begin{eqnarray}\label{eq_JDRSAM_assumption_eq_beta}
    \beta^k(t) = \theta B^{-1}\left(1-e^{B(T-t)}\right)\left(A_0+\bar{h}^k \hat{A}\right)
\end{eqnarray}
and $2n$ functions $\alpha^k: [0,T] \to \mathbb{R}$ defined by
\begin{eqnarray}\label{eq_JDRSAM_assumption_eq_alpha}
    &&  \alpha(t) = -\int_{t}^{T} q(s)ds
\end{eqnarray}
where
\begin{eqnarray}
    q(t)
    &:=&    \left(b- \theta \Lambda \Sigma'\bar{h}
        +   \int_{\mathbf{Z}}\xi(z)\left[
                \left( 1+\bar{h}^{k'}\gamma(z)\right)^{-\theta}
                    - \mathit{1}_{\mathbf{Z}_0}(z)
                \right]\nu(dz)
            \right)' \beta^{k'}(t)
                                       \nonumber\\
    &&  + \frac{1}{2} \textrm{tr} \left( \Lambda \Lambda' \beta^{k'}(t)\beta^k(t) \right)
        + \int_{\mathbf{Z}} \left\{
            e^{\beta^k \xi(z)}
                - 1
                - \xi'(z) \beta^{k'}(t)
            \right\} \nu(dz)
                                       \nonumber\\
    &&  + \frac{1}{2} \theta \left(\theta+1 \right)\bar{h}^{k'}\Sigma\Sigma'\bar{h}^{k}
            - \theta a_0
            - \theta \hat{a}
        + \theta \int_{\mathbf{Z}}  \left\{\frac{1}{\theta}
                \left[\left(1+\bar{h}^{k'}\gamma(z)\right)^{-\theta}-1\right]
                +\bar{h}^{k'}\gamma(z)\mathit{1}_{\mathbf{Z}_0}(z)
                \right\} \nu(dz)
                                        \nonumber
\end{eqnarray}
exist and for $i=1,\ldots,n$ satisfy:
\begin{eqnarray}\label{eq_JDRSAM_assumption_cond_beta}
    \beta_{i}^{i}(t) &<& 0
                                            \nonumber\\
    \beta_{i}^{n+i}(t)   &>& 0
\end{eqnarray}
where $\beta_{j}^{i}(t)$ denotes the $j$-th component of the vector $\beta^i(t)$.

\end{assumption}

\begin{remark}\label{rk_JDRSAM_beta_structure}
    Key to this assumption is the condition~\eqref{eq_JDRSAM_assumption_cond_beta} which imposes a specific constraint on one element of each of the $2n$ vectors $\beta^{k}(t)$. To clarify the structure of this constraint, define $M_{\beta}^{-}$ as the square $n \times n$ matrix whose $i$-th column (with $i =1,\ldots n$) is the $n$-element column vector $\beta^i(t)$. Then all the elements $m_{jj}^{-}, j = 1,\ldots,m$ on the diagonal of $M_{\beta}^{-}$ are such that
    \begin{equation}
        m_{jj}^{-} = \beta_{j}^{j}(t) < 0
                                                        \nonumber
    \end{equation}

    Similarly, define $M_{\beta}^{+}$ as the square $n \times n$ matrix whose $i$-th column (with $i =1,\ldots n$) is the $n$-element column vector $\beta^{n+i}(t)$. Then all the elements $m_{jj}^{+}, j = 1,\ldots,m$ on the diagonal of $M_{\beta}^{+}$ are such that
    \begin{equation}
        m_{jj}^{+} = \beta_{j}^{n+j}(t) > 0
                                                        \nonumber
    \end{equation}

    Note that there is no requirement for either $M_{\beta}^{-}$ or $M_{\beta}^{+}$ to have full rank. It would in fact be perfectly acceptable to have rank 1 as a result of column duplication.

\end{remark}

\begin{remark}\label{rk_JDRSAM_beta_existence}
    For the function $\beta^k$ in equation~\eqref{eq_JDRSAM_assumption_eq_beta} to exists, $B$ must be invertible. Moreover, the existence of $2n$ constant controls $\bar{h}^k, k=1,\ldots,2n$ such that~\eqref{eq_JDRSAM_assumption_eq_beta} satisfies~\eqref{eq_JDRSAM_assumption_cond_beta} is only guaranteed when $J= \mathbb{R}^n$. However, since finding the controls is equivalent to solving a system of at most $n$ inequalities with $m$ variables and $m>n$, it is likely that one could find constant controls after some adjustments to the elements of the matrices $A_0, A, B$ or to the maximum jump size allowed.
\end{remark}

%%%%%%%%%%%%%%%%%%%%%%%%%%%%%%%%%%%%%%%%%%%%%%%%%%%%%%%%%%%%%%%%%%%%%%%%%%%%%%
%   Proposition
%%%%%%%%%%%%%%%%%%%%%%%%%%%%%%%%%%%%%%%%%%%%%%%%%%%%%%%%%%%%%%%%%%%%%%%%%%%%%%

\begin{proposition}\label{prop_JDRSAM_behaviour_tildePhi}
Suppose Assumption~\ref{as_JDRSAM_ODEs_alpha_beta} holds and consider the $2n$ constant controls $\bar{h}^k, k=1,\ldots,2n$ parameterizing the $4n$ functions
\begin{equation}
	\alpha^k: [0,T] \to \mathbb{R},       \;  k=1,\ldots,2n
                                            \nonumber
\end{equation}
\begin{equation}
   \beta^k: [0,T] \to \mathbb{R}^n,       \;  k=1,\ldots,2n
                                            \nonumber
\end{equation}
such that for $i=1,\ldots,n$,
\begin{eqnarray}
	\beta_{i}^{i}(t) &<& 0
                                            \nonumber\\
   \beta_{i}^{n+i}(t)   &>& 0
                                            \nonumber
\end{eqnarray}
where $\beta_{j}^{i}(t)$ denotes the $j$-th component of the vector $\beta^i(t)$. Then we have the following upper bounds:
    \begin{eqnarray}
            \tilde{\Phi}(t,x) \leq e^{\alpha^k(t) + \beta^{k'}(t)x}
                                            \nonumber
    \end{eqnarray}
    in each element $x_i, i=1,\ldots,n$ of $x$.

\end{proposition}

%%%%%%%%%%%%%%%%%%%%%%%%%%%%%%%%%%%%%%%%%%%%%%%%%%%%%%%%%%%%%%%%%%%%%%%%%%%%%%
%   PROOF OF PROPOSITION
%%%%%%%%%%%%%%%%%%%%%%%%%%%%%%%%%%%%%%%%%%%%%%%%%%%%%%%%%%%%%%%%%%%%%%%%%%%%%%

\begin{proof}
Setting $\mathbf{Z} = \mathbb{R}^n-\left\{0\right\}$ and recalling that the dynamics of the state variable $X(t)$ under the $\mathbb{P}_{h}^{\theta}$-measure is given by
\begin{eqnarray}
    dX(t)
    &=&     f(X(t^-),h(t);\theta)
				+ \Lambda dW_{t}^{h}
            + \int_{\mathbb{R}^n}\xi(z)\tilde{N}_{\textbf{p}}^{h}(dt,dz)
                                                            \nonumber
\end{eqnarray}
we note that the associated L\'evy measure $\tilde{\nu}$ can be defined via the map:
\begin{equation}\label{eq_JDRSAM_measure_tildenu_def}
    \tilde{\nu} = \nu \circ \xi^{-1}
\end{equation}

We will now limit ourselves the class $\mathcal{H}^c$ of constant controls.
By the optimality principle, for an arbitrary admissible constant control
policy $\bar{h}$, we have
\begin{equation}\label{eq_JDRSAM_prop_tildePhi_eq1}
    \tilde{\Phi}(t,x) \leq \tilde{I}(x;\bar{h};t,T) \leq \mathbf{E}_{t,x} \left[ \exp \left\{ \theta \int_{t}^{T} g(X_s,\bar{h}) ds -\theta\ln v \right\}\right] := W(t,x)
\end{equation}
\\

In this setting, we note that the function $g$ is an affine function of the affine process $X(t)$. Affine process theory See Appendix A in Duffie and Singleton~\cite{dusi03}, Duffie, Pan and Singleton~\cite{dupasi00} or Duffie, Filipovic and Schachermayer~\cite{dufisc03} for more details on the properties of affine processes) leads us to expect that the expectation on the right-hand side of equation~\eqref{eq_JDRSAM_prop_tildePhi_eq1} takes the form
\begin{equation}\label{eq_JDRSAM_prop_tildePhi_eq2}
    W(t,x) =   \exp \left\{ \alpha(t) + \beta(t)x\right\}
\end{equation}
where
\begin{equation}
    \alpha: t \in [0,T] \to \mathbb{R}
                                            \nonumber
\end{equation}
\begin{equation}
    \beta: t \in [0,T] \to \mathbb{R}^n
                                            \nonumber
\end{equation}
are functions solving two ODEs.
\\

Indeed, applying the Feynman-Kac formula, we find that the function $W(t,x)$ satisfies the integro-differential PDE:
\begin{eqnarray}
    &&  \frac{\partial W}{\partial t}
        +   \left(b+ B X_s - \theta \Lambda \Sigma'\bar{h}
        +   \int_{\mathbf{Z}}\xi(z)\left[
              \left( 1+\bar{h}'\gamma(z)\right)^{-\theta}
                - \mathit{1}_{\mathbf{Z}_0}(z)
            \right]\nu(dz)
         \right)' DW(t,x)
                                        \nonumber\\
    &&  + \frac{1}{2} \textrm{tr} \left( \Lambda \Lambda' D^2 W(t,x)\right)
        + \int_{\mathbf{Z}} \left\{
              W(t,x+\xi(z))
            - W(t,x)
            - \xi'(z) DW(t,x)
        \right\} \nu(dz)
                                       \nonumber\\
    &&  + \theta g(x,\bar{h};\theta) W(t,x)
                                       \nonumber\\
    &=& 0
                                        \nonumber
\end{eqnarray}
subject to terminal condition $\tilde{\Phi}(T, x) = v^{-\theta}$.
\\

Now, taking a candidate solution of the form
\begin{equation}
    W(t,x) =   \exp \left\{ \alpha(t) + \beta(t)x\right\}
                                            \nonumber
\end{equation}
we have
\begin{eqnarray}
    \frac{\partial W}{\partial t} &=& \left(\dot{\alpha(t)} + \dot{\beta}(t)x \right) W(t,x)
                                            \nonumber\\
    DW  &=& \beta'(t) W(t,x)
                                            \nonumber\\
    D^2 W   &=&\beta'(t)\beta(t) W(t,x)
                                            \nonumber
\end{eqnarray}
Substituting into the PDE, we get
\begin{eqnarray}
    &&  \left(\dot{\alpha(t)} + \dot{\beta}(t)x \right) W(t,x)
                                        \nonumber\\
    &&  +   \left(b+ B x - \theta \Lambda \Sigma'\bar{h}
        +   \int_{\mathbf{Z}}\xi(z)\left[
              \left( 1+h'\gamma(z)\right)^{-\theta}
                - \mathit{1}_{\mathbf{Z}_0}(z)
            \right]\nu(dz)
         \right)' \beta'(t) W(t,x)
                                        \nonumber\\
    &&  + \frac{1}{2} \textrm{tr} \left( \Lambda \Lambda' \beta'(t)\beta(t) \right)W(t,x)
                                        \nonumber\\
    &&  + \int_{\mathbf{Z}} \left\{
              W(t,x+\xi(z))
            - W(t,x)
            - \xi'(z) \beta'(t) W(t,x)
        \right\} \nu(dz)
                                       \nonumber\\
    &&  + \theta \left(
            \frac{1}{2} \left(\theta+1 \right)\bar{h}'\Sigma\Sigma'\bar{h}- a_0-A_0'x -\bar{h}'(\hat{a} + \hat{A}x)
            \right.
                                        \nonumber\\
    && \left.
        + \int_{\mathbf{Z}}  \left\{\frac{1}{\theta}
            \left[\left(1+h'\gamma(z)\right)^{-\theta}-1\right]
            +\bar{h}'\gamma(z)\mathit{1}_{\mathbf{Z}_0}(z)
            \right\} \nu(dz)
        \right) W(t,x)
                                        \nonumber\\
    &=& 0
                                        \nonumber
\end{eqnarray}

Dividing by $W(t,x)$ and rearranging, we get
\begin{eqnarray}
    &&  \left(
            \dot{\beta}(t)
            + B'\beta'(t)
            - \theta A_0'
            -\theta \bar{h}' \hat{A}
        \right) x
                                        \nonumber\\
    &=& -\left(
            \dot{\alpha(t)}
            +   \left(b- \theta \Lambda \Sigma'\bar{h}
            +   \int_{\mathbf{Z}}\xi(z)\left[
                    \left( 1+\bar{h}'\gamma(z)\right)^{-\theta}
                        - \mathit{1}_{\mathbf{Z}_0}(z)
                    \right]\nu(dz)
                \right)' \beta'(t)
        \right.
                                       \nonumber\\
    &&  \left.
            + \frac{1}{2} \textrm{tr} \left( \Lambda \Lambda' \beta'(t)\beta(t) \right)
            + \int_{\mathbf{Z}} \left\{
              e^{\beta \xi(z)}
                - 1
                - \xi'(z) \beta'(t)
        \right\} \nu(dz)
        \right.
                                       \nonumber\\
    &&  \left.
            + \frac{1}{2} \theta \left(\theta+1 \right)\bar{h}'\Sigma\Sigma'\bar{h}
            - \theta a_0
            - \theta \hat{a}
            + \theta \int_{\mathbf{Z}}  \left\{\frac{1}{\theta}
                \left[\left(1+\bar{h}'\gamma(z)\right)^{-\theta}-1\right]
                +\bar{h}'\gamma(z)\mathit{1}_{\mathbf{Z}_0}(z)
                \right\} \nu(dz)
        \right)
                                        \nonumber
\end{eqnarray}

Since the left-hand side is independent from the right-hand side, then both sides are orthogonal. As a result we now only need to solve the two ODEs
\begin{eqnarray}\label{eq_JDRSAM_prop_tildePhi_ODE_beta}
    \dot{\beta}(t)
    + B'\beta'(t)
    - \theta A_0'
    - \theta \bar{h}' \hat{A}
    = 0
\end{eqnarray}
and
\begin{eqnarray}\label{eq_JDRSAM_prop_tildePhi_ODE_alpha}
    &&  \dot{\alpha(t)}
        +   \left(b- \theta \Lambda \Sigma'\bar{h}
        +   \int_{\mathbf{Z}}\xi(z)\left[
                \left( 1+\bar{h}'\gamma(z)\right)^{-\theta}
                    - \mathit{1}_{\mathbf{Z}_0}(z)
                \right]\nu(dz)
            \right)' \beta'(t)
                                       \nonumber\\
    &&  + \frac{1}{2} \textrm{tr} \left( \Lambda \Lambda' \beta'(t)\beta(t) \right)
        + \int_{\mathbf{Z}} \left\{
              e^{\beta \xi(z)}
                - 1
                - \xi'(z) \beta'(t)
        \right\} \nu(dz)
                                       \nonumber\\
    &&  + \frac{1}{2} \theta \left(\theta+1 \right)\bar{h}'\Sigma\Sigma'\bar{h}
            - \theta a_0
            - \theta \hat{a}
        + \theta \int_{\mathbf{Z}}  \left\{\frac{1}{\theta}
                \left[\left(1+\bar{h}'\gamma(z)\right)^{-\theta}-1\right]
                +\bar{h}'\gamma(z)\mathit{1}_{\mathbf{Z}_0}(z)
                \right\} \nu(dz)
                                        \nonumber\\
    &=& 0
\end{eqnarray}
to obtain the value of $W(t,x)$. The ODE~\eqref{eq_JDRSAM_prop_tildePhi_ODE_beta} for $\beta$ is linear and admits the solution
\begin{eqnarray}\label{eq_JDRSAM_prop_tildePhi_eq_beta}
    \beta(t) = \theta B^{-1}\left(1-e^{B(T-t)}\right)\left(A_0+\bar{h}^k \hat{A}\right)
\end{eqnarray}
As for the ODE~\eqref{eq_JDRSAM_prop_tildePhi_ODE_alpha} for $\alpha$, we only need to integrate to get
\begin{eqnarray}\label{eq_JDRSAM_prop_tildePhi_eq_alpha}
    &&  \alpha(t) = -\int_{t}^{T} q(s)ds
\end{eqnarray}
where
\begin{eqnarray}
    q(t)
    &:=&    \left(b- \theta \Lambda \Sigma'\bar{h}
        +   \int_{\mathbf{Z}}\xi(z)\left[
                \left( 1+\bar{h}'\gamma(z)\right)^{-\theta}
                    - \mathit{1}_{\mathbf{Z}_0}(z)
                \right]\nu(dz)
            \right)' \beta'(t)
                                       \nonumber\\
    &&  + \frac{1}{2} \textrm{tr} \left( \Lambda \Lambda' \beta'(t)\beta(t) \right)
        + \int_{\mathbf{Z}} \left\{
            e^{\beta \xi(z)}
                - 1
                - \xi'(z) \beta'(t)
            \right\} \nu(dz)
                                       \nonumber\\
    &&  + \frac{1}{2} \theta \left(\theta+1 \right)\bar{h}'\Sigma\Sigma'\bar{h}
            - \theta a_0
            - \theta \hat{a}
        + \theta \int_{\mathbf{Z}}  \left\{\frac{1}{\theta}
                \left[\left(1+\bar{h}'\gamma(z)\right)^{-\theta}-1\right]
                +\bar{h}'\gamma(z)\mathit{1}_{\mathbf{Z}_0}(z)
                \right\} \nu(dz)
                                        \nonumber
\end{eqnarray}
\\

Observe that $W(t,x)$ is increasing in $x_i$, the $i$-th element of $x$, if $\beta_i > 0$, and conversely, $W(t,x)$ is decreasing in $x_i$ if $\beta_i < 0$
\\

Equations~\eqref{eq_JDRSAM_prop_tildePhi_eq_beta} and~\eqref{eq_JDRSAM_prop_tildePhi_eq_alpha} are respectively equations~\eqref{eq_JDRSAM_assumption_eq_beta} and~\eqref{eq_JDRSAM_assumption_eq_alpha} from Assumption~\ref{as_JDRSAM_ODEs_alpha_beta}. By Assumption~\ref{as_JDRSAM_ODEs_alpha_beta}, there exists $2n$ constant controls $\bar{h}^k, k=1,\ldots,2n$ such that for $i=1,\ldots,n$,
\begin{eqnarray}
    \beta_{i}^{i}(t) &<& 0
                                            \nonumber\\
    \beta_{i}^{n+i}(t)   &>& 0
                                            \nonumber
\end{eqnarray}
where $\beta_{j}^{i}(t)$ denotes the $j$-th component of the vector $\beta^i(t)$. We can now conclude that we have the following upper bounds
\begin{eqnarray}
    \tilde{\Phi}(t,x) \leq e^{\alpha^k(t) + \beta^{k'}(t)x}
                                            \nonumber
\end{eqnarray}
for each element $x_i, i=1,\ldots,n$ of $x$.

\end{proof}

%%%%%%%%%%%%%%%%%%%%%%%%%%%%%%%%%%%%%%%%%%%%%%%%%%%%%%%%%%%%%%%%%%%%%%%%%%%%%%
%   END OF PROOF
%%%%%%%%%%%%%%%%%%%%%%%%%%%%%%%%%%%%%%%%%%%%%%%%%%%%%%%%%%%%%%%%%%%%%%%%%%%%%%

\begin{remark}
    To obtain the upper bounds and the asymptotic behaviour, we do not need the $2n$ constant controls to be pairwise different. In fact, we need at least $2$ different controls and at most $2n$ different controls. Moreover, we could consider wider classes of controls extending beyond constant controls. This would require some modifications to the proof but would also alleviate the assumptions required for the result to hold.
\\
\end{remark}

\begin{remark}
    For a given constant control $\bar{h}$, equation~\eqref{eq_JDRSAM_prop_tildePhi_ODE_beta} is a linear $n$-dimensional ODE. However, if in the dynamics of the state variable $X(t)$, $\Lambda$ and $\Xi$ depended on $X$, the ODE would be nonlinear. Once ODE~\eqref{eq_JDRSAM_prop_tildePhi_ODE_beta} is solved, obtaining $\alpha(t)$ from equation~\eqref{eq_JDRSAM_prop_tildePhi_ODE_alpha} is a simple matter of integration.
\\
\end{remark}

\begin{remark}
    For a given constant control $h$, given $x \in \mathbb{R}^n$ and $t \in [0,T]$, the solution of ODE~\eqref{eq_JDRSAM_prop_tildePhi_ODE_beta} is the same whether the dynamics of $S(t)$ and $X(t)$ is the jump diffusion considered here or the corresponding pure diffusion model. The converse is, however, not true since in the pure diffusion setting $h \in \mathbb{R}^m$, while in the jump diffusion case $h \in \mathcal{J} \subset \mathbb{R}^m$.
\end{remark}

%%%%%%%%%%%%%%%%%%%%%%%%%%%%%%%%%%%%%%%%%%%%%%%%%%%%%%%%%%%%%%%%%%%%%%%%%%%%%%%%%%%
%
%   NEXT SECTION
%
%%%%%%%%%%%%%%%%%%%%%%%%%%%%%%%%%%%%%%%%%%%%%%%%%%%%%%%%%%%%%%%%%%%%%%%%%%%%%%%%%%%

%%%
\section{Viscosity Solution Approach}
In recent years, viscosity solutions have gained a widespread acceptance as an effective technique to obtain a weak sense solution for HJB PDEs when no classical (i.e $C^{1,2}$) solution can be shown to exist, which is the case for many stochastic control
problems. Viscosity solutions also have a very practical interest. Indeed, once a solution has been interpreted in the viscosity sense and the
uniqueness of this solution has been proved via a comparison result, the fundamental `stability' result of Barles and Souganidis~\cite{baso91} opens the way to a numerical resolution of the problem through a wide range of schemes. Readers interested in an overview of viscosity solutions should refer to the classic article by Crandall, Ishii and Lions~\cite{crisli92}, the book by Fleming and Soner~\cite{flso06} and \O ksendal and Sulem~\cite{oksu05}, as well as the notes by Barles~\cite{ba97} and Touzi~\cite{to02}.
\\

While the use of viscosity solutions to solve classical diffusion-type stochastic
control problems has been extensively studied and surveyed (see Fleming and Soner~\cite{flso06} and Touzi~\cite{to02}), this introduction of a jump-related measure makes the jump-diffusion framework more complex. As a result, so far no general theory has been developed to solve jump-diffusion problems. Instead, the assumptions made to derive a comparison result are closely related to what the specific problem allows. Broadly speaking, the literature can be split along two lines of analysis, depending on whether the measure associated with the jumps is assumed to be finite.
\\

In the case when the jump measure is finite, Alvarez and Tourin~\cite{alto96} consider a fairly general setting in which the jump term does not need to be linear in the function $u$ which solves the integro-differential PDE. In this setting, Alvarez and Tourin develop a comparison theorem that they apply to a stochastic differential utility problem. Amadori~\cite{am03} extends Alvarez and Tourin's analysis to price European options. Barles, Buckdahn and Pardoux~\cite{babupa97} study the viscosity solution of integro-differential equations associated with backward SDEs (BSDEs).
\\

The L\'evy measure is the most extensively studied measure with singularities. Pham~\cite{ph98} derives a comparison result for the variational inequality associated with an optimal stopping problem. Jakobsen and Karlsen~\cite{jaka06} analyse in detail the impact of the L\'evy measure's singularity and propose a maximum principle. Amadori, Karlsen and La Chioma~\cite{amkalc04} focus on geometric L\'evy processes and the partial integro differential equations they generate before applying their results to BSDEs and to the pricing of European and American derivatives. A recent article by Barles and Imbert~\cite{baim08} takes a broader view of PDEs and their non-local operators. However, the authors assume that the nonlocal operator is broadly speaking linear in the solution which may prove overly restrictive in some cases, including our present problem. 
\\

As far as our jump diffusion risk-sensitive control problem is concerned, we will promote a general treatment and avoid restricting the class of the compensator $\nu$. At some point, we will however need $\nu$ to be finite. This assumption will only be made for a purely technical reason arising in the proof of the comparison result (in Section 6). Since the rest of the story is still valid if $\nu$ is not finite, and in accordance with our goal of keeping the discussion as broad as possible, we will write the rest of the article in the spirit of a general compensator $\nu$.
\\

\subsection{Definitions}
Before proceeding further, we will introduce the following definition:

\begin{definition}
    The upper semicontinuous envelope $u^*(x)$ of a function $u$ at $x$ is defined as
\begin{equation}
    u^*(x) = \limsup_{y \to x} u(y)
                                            \nonumber
\end{equation}
and the lower semicontinuous envelope $u_*(x)$ of $u(x)$ is defined as
\begin{equation}
    u_*(x) = \liminf_{y \to x} u(y)
                                            \nonumber
\end{equation}
\end{definition}

Note in particular the fundamental inequality between a function and its upper and lower semicontinuous envelopes:
\begin{equation}
    u_*     \leq    u   \leq    u^*
                                            \nonumber
\end{equation}
\\

The theory of viscosity solutions was initially developed for elliptical PDEs of the form
\begin{equation}
    H(x,u,Du,D^2u) = 0
                                            \nonumber
\end{equation}
and parabolic PDEs of the form
\begin{equation}\label{eq_JDRSAM_viscosity_general parabolic_PDE}
    \frac{\partial u}{\partial t} + H(x,u,Du,D^2u) = 0
                                            \nonumber
\end{equation}
for what Crandall, Ishii and Lions~\cite{crisli92} term a ``proper'' functional $H(x,r,p,A)$.

\begin{definition}
A functional $H(x,r,p,A)$ is said to be \emph{proper} if it satisfies the following two properties:
\begin{enumerate}
    \item (degenerate) ellipticity:
    \begin{equation}
        H(x,r,p,A) \leq H(x,r,p,B),
        \qquad B \leq A
                                    \nonumber
    \end{equation}
    and
    \item monotonicity
    \begin{equation}
        H(x,r,p,A) \leq H(x,s,p,A),
        \qquad r \leq s
                                    \nonumber
    \end{equation}
\end{enumerate}

\end{definition}

In our problem, the functional $F$ defined as
\begin{eqnarray}\label{eq_JDRSAM_HJBPDE_def_F_functional}
F(x,p,A)
    &:=&
   -    \sup_{h \in \mathcal{J}}\left\{
        f(x,h)'p
   +    \frac{1}{2} \textrm{tr}\left(\Lambda \Lambda' A\right)
    \right.
                                                    \nonumber\\
    &&  \left.
   -    \frac{\theta}{2} p'\Lambda\Lambda'p
    \right.
                                                    \nonumber\\
    &&  \left.
    +   \int_{\mathbf{Z}} \left\{
        - \frac{1}{\theta}\left(
                e^{-\theta \left(\Phi(t,x+\xi(z)) - \Phi(t,x) \right)} -1
            \right)
        - \xi'(z) p
        \right\} \nu(dz)
    \right.
                                                    \nonumber\\
    &&  \left.
   -    g(x,h)
                                \right\}
\end{eqnarray}
plays a similar role to the functional $H$ in the general equation~\eqref{eq_JDRSAM_viscosity_general parabolic_PDE}, and we note that it is indeed ``proper''. As a result, we can develop a viscosity approach to show that the value function $\Phi$ is the unique solution of the associated RS HJB PIDE.
\\

We now give two equivalent definitions of viscosity solutions adapted from Alvarez and Tourin~\cite{alto96}:
\begin{itemize}
\item a definition based on the notion of semijets;
\item a definition based on the notion of test function
\end{itemize}
Before introducing these two definitions, we need to define parabolic semijet of upper semicontinuous and lower semicontinuous functions and to add two additional conditions.
\\

\begin{definition}[Parabolic Semijets]
Let $u \in USC([0,T]\times\mathbb{R}^n)$ and $(t,x) \in [0,T]\times\mathbb{R}^n$. We define:
\begin{itemize}
\item the Parabolic superjet $\mathcal{P}_u^{2,+}$ as
\begin{eqnarray}
	\mathcal{P}_u^{2,+}
	&:=&	\left\{(p,q,A) \in \mathbb{R} \times \mathbb{R}^n \times \mathcal{S}_n:
			\right.
																\nonumber\\
	&&	\left.
			u(s,y) \leq u(s,x)
			+ p(s-t)
			+\left<q, y-x \right>
			+\frac{1}{2}\left<A(y-x), y-x \right>
			\right.
																\nonumber\\
	&&	\left.
			+ o(\left|s-t\right| + \left|y-x\right|^2)
			\textrm{ as } (s,y) \to (t,x)
	\right\}
																\nonumber
\end{eqnarray}
\item the closure of the Parabolic superjet $\overline{\mathcal{P}}_u^{2,+}$ as
\begin{eqnarray}
	\overline{\mathcal{P}}_u^{2,+}
	&:=&	\left\{ (p,q,A) = \lim_{k\to\infty}(p_k,q_k,A_k)
		\textrm{ with } (p_k,q_k,A_k) \in \mathcal{P}_u^{2,+}
			\right.
																\nonumber\\
	&&	\left.
		\textrm{ and } \lim_{k\to\infty}(t_k,x_k,u(t_k,x_k)) = (t,x,u(t,x))
	\right\}
																\nonumber
\end{eqnarray}
\end{itemize}

Let $u \in LSC([0,T]\times\mathbb{R}^n)$ and $(t,x) \in [0,T]\times\mathbb{R}^n$. We define:
\begin{itemize}
\item the Parabolic subjet $\mathcal{P}_u^{2,-}$ as $\mathcal{P}_u^{2,-} := - \mathcal{P}_u^{2,+}$, and;
\item the closure of the Parabolic subjet $\overline{\mathcal{P}}_u^{2,-}$ as $\overline{\mathcal{P}}_u^{2,-} = -\overline{\mathcal{P}}_u^{2,+}$
\end{itemize}

\end{definition}

\begin{condition}[Condition on an Upper Semicontinuous Function $u$]\label{cond_JDRSAM_viscosity_uppersemicont}
Let $(t,x) \in [0,T]\times\mathbb{R}^n$ and $(p,q,A) \in \mathcal{P}^{2,+}u(t,x)$, there are $\varphi \in C(\mathbb{R}^n)$, $\varphi \geq 1$ and $R > 0$ such that for
\begin{eqnarray}
	\left((s,y),z\right) \in \left(\mathscr{B}_{R}(t,x)\cap\left([0,T]\times\mathbb{R}^n\right)\right)\times\mathbf{Z},
															\nonumber
\end{eqnarray}
\begin{eqnarray}
	\int_{\mathbf{Z}} \left\{
        - \frac{1}{\theta}\left(
                e^{-\theta \left(u(s,y+\xi(z)) - u(s,y) \right)} -1
            \right)
        - \xi'(z) q
        \right\} \nu(dz)
	\leq \varphi(y)
															\nonumber
\end{eqnarray}
\end{condition}

\begin{condition}[Condition on a Lower Semicontinuous Function $u$]\label{cond_JDRSAM_viscosity_lowersemicont}
Let $(t,x) \in [0,T]\times\mathbb{R}^n$ and $(p,q,A) \in \mathcal{P}^{2,-}u(t,x)$, there are $\varphi \in C(\mathbb{R}^n)$, $\varphi \geq 1$ and $R > 0$ such that for
\begin{eqnarray}
	\left((s,y),z\right) \in \left(\mathscr{B}_{R}(t,x)\cap\left([0,T]\times\mathbb{R}^n\right)\right)\times\mathbf{Z},
															\nonumber
\end{eqnarray}
\begin{eqnarray}
	\int_{\mathbf{Z}} \left\{
        - \frac{1}{\theta}\left(
                e^{-\theta \left(u(s,y+\xi(z)) - u(s,y) \right)} -1
            \right)
        - \xi'(z) q
        \right\} \nu(dz)
	\geq -\varphi(y)
															\nonumber
\end{eqnarray}
\end{condition}

The purpose of these conditions on $u$ and $v$ is to ensure that the jump term is semicontinuous at any given point $(t,x) \in [0,T]\times\mathbb{R}^n$ (see Lemma 1 and Conditions (6) and (7) in~\cite{alto96}). In our setting, we note that since the value function $\Phi$ and the function $x \mapsto e^x$ are locally bounded, these two conditions are satisfied.
\\

\begin{remark}
Note that the jump-related integral term
\begin{equation}
	\int_{\mathbf{Z}} \left\{
        - \frac{1}{\theta}\left(
                e^{-\theta \left(u(s,y+\xi(z)) - u(s,y) \right)} -1
            \right)
        - \xi'(z) q
        \right\} \nu(dz)
                                                                \nonumber
\end{equation}
is well defined when $(p,q,A) \in \mathcal{P}_u^{2,\pm}$. First, by Taylor,
\begin{eqnarray}
    &&  \int_{\mathbf{Z}} \left\{
        - \frac{1}{\theta}\left(
                e^{-\theta \left(u(s,y+\xi(z)) - u(s,y) \right)} -1
            \right)
        - \xi'(z) q
        \right\} \nu(dz)
                                                                \nonumber\\
    &=& \int_{\mathbf{Z}} \left\{
            \left(u(s,y+\xi(z)) - u(s,y) \right)
        -   \frac{\theta}{2} \left(u(s,y+\xi(z)) - u(s,y) \right)^2
        \right.
                                                                \nonumber\\
    &&  \left.
        +   \frac{\theta^2}{3!} \left(u(s,y+\xi(z)) - u(s,y) \right)^3
        +   \ldots
        -   \xi'(z)q
        \right\} \nu(dz)
                                                                \nonumber
\end{eqnarray}
By definition of the Parabolic superjet $\mathcal{P}_u^{2,+}$, for $t=s$, the pair $(q,A)$ satisfies the inequality
\begin{eqnarray}
        u(s,y+\xi(z))
    -   u(s,y)
    -   \xi'(z)q
    \leq
	   \frac{1}{2}\xi'(z)A\xi(z)
	+  o(\left|\xi(z)\right|^2)
                                                                \nonumber
\end{eqnarray}
Similarly, by definition of the Parabolic subjet $\mathcal{P}_u^{2,-}$, for $t=s$, the pair $(q,A)$ satisfies the inequality
\begin{eqnarray}
        u(s,y+\xi(z))
    -   u(s,y)
    -   \xi'(z)q
    \geq
	   \frac{1}{2}\xi'(z)A\xi(z)
	+  o(\left|\xi(z)\right|^2)
                                                                \nonumber
\end{eqnarray}
Thus, if $u$ is a viscosity solution, we have
\begin{eqnarray}
        u(s,y+\xi(z))
    -   u(s,y)
    -   \xi'(z)q
    =
	   \frac{1}{2}\xi'(z)A\xi(z)
	+  o(\left|\xi(z)\right|^2)
                                                                \nonumber
\end{eqnarray}
and the jump-related integral is equal to
\begin{eqnarray}
    &&  \int_{\mathbf{Z}} \left\{
        - \frac{1}{\theta}\left(
                e^{-\theta \left(u(s,y+\xi(z)) - u(s,y) \right)} -1
            \right)
        - \xi'(z) q
        \right\} \nu(dz)
                                                                \nonumber\\
    &=& \int_{\mathbf{Z}} \left\{
        -   \frac{\theta}{2} \left(u(s,y+\xi(z)) - u(s,y) \right)^2
        +   \frac{1}{2}\xi'(z)A\xi(z)
    	+  o(\left|\xi(z)\right|^2)
        \right\} \nu(dz)
                                                                \nonumber
\end{eqnarray}
which is well-defined.

\end{remark}

\begin{definition}[Viscosity Solution (Semijets)]\label{def_JDRSAM_viscositysol_semijet}
A locally bounded function $u \in USC([0,T]\times\mathbb{R}^n)$ satisfying Condition~\ref{cond_JDRSAM_viscosity_uppersemicont} is a viscosity subsolution of~\eqref{eq_JDRSAM_HJBPDE}, if for all $x \in \mathbb{R}^n$, $u(T,x) \leq g_0(x)$, and for all $(t,x) \in [0,T]\times\mathbb{R}^n$, $(p,q,A) \in \mathcal{P}^{2,+}u(t,x)$, we have
\begin{equation}
	- p
	+ F(x,q,A)
	- \int_{\mathbf{Z}} \left\{
   	- \frac{1}{\theta}\left(
      		e^{-\theta \left(u(t,x+\xi(z)) - u(t,x) \right)} -1
      	\right)
    	- \xi'(z) q \right\}
      \nu(dz)
	\leq 0
																\nonumber
\end{equation}
\\

A locally bounded function $u \in LSC([0,T]\times\mathbb{R}^n)$ satisfying Condition~\ref{cond_JDRSAM_viscosity_lowersemicont} is a viscosity supersolution of~\eqref{eq_JDRSAM_HJBPDE}, if for all $x \in \mathbb{R}^n$, $u(T,x) \geq g_0(x)$, and for all $(t,x) \in [0,T]\times\mathbb{R}^n$, $(p,q,A) \in \mathcal{P}^{2,-}u(t,x)$, we have
\begin{equation}
	- p
	+ F(x,q,A)
	- \int_{\mathbf{Z}} \left\{
   	- \frac{1}{\theta}\left(
      		e^{-\theta \left(u(t,x+\xi(z)) - u(t,x) \right)} -1
      	\right)
    	- \xi'(z) q \right\}
      \nu(dz)
	\geq 0
																\nonumber
\end{equation}
\\

A locally bounded function $\Phi$ whose upper semicontinuous and lowersemicontinuous envelopes are a viscosity subsolution and a viscosity supersolution of~\eqref{eq_JDRSAM_HJBPDE} is a viscosity solution of~\eqref{eq_JDRSAM_HJBPDE}.

\end{definition}

\begin{definition}[Viscosity Solution (Test Functions)]\label{def_JDRSAM_viscositysol_testfunc}
A locally bounded function $u \in USC([0,T]\times\mathbb{R}^n)$ is a viscosity subsolution of~\eqref{eq_JDRSAM_HJBPDE}, if for all $x \in \mathbb{R}^n$, $u(T,x) \leq g_0(x)$, and for all $(t,x) \in [0,T]\times\mathbb{R}^n$, $\psi \in C^2([0,T]\times\mathbb{R}^n)$ such that $u(t,x) = \psi(t,x)$, $u < \psi$ on $[0,T]\times\mathbb{R}^n \backslash \left\{(t,x)\right\}$, we have
\begin{equation}
	-  \frac{\partial \psi}{\partial t}
	+  F(x,D\psi,D^2\psi)
	-  \int_{\mathbf{Z}} \left\{
    -  \frac{1}{\theta}\left(
      		e^{-\theta \left(\psi(t,x+\xi(z)) - \psi(t,x) \right)} -1
      	\right)
    	- \xi'(z) D\psi \right\}
      \nu(dz)
	\leq 0
																\nonumber
\end{equation}
\\

A locally bounded function $v \in LSC([0,T]\times\mathbb{R}^n)$ is a viscosity supersolution of~\eqref{eq_JDRSAM_HJBPDE}, if for all $x \in \mathbb{R}^n$, $v(T,x) \geq g_0(x)$, and for all $(t,x) \in [0,T]\times\mathbb{R}^n$, $\psi \in C^2([0,T]\times\mathbb{R}^n)$ such that $v(t,x) = \psi(t,x)$, $v > \psi$ on $[0,T]\times\mathbb{R}^n \backslash \left\{(t,x)\right\}$, we have
\begin{equation}
	-  \frac{\partial \psi}{\partial t}
	+  F(x,D\psi,D^2\psi)
	-  \int_{\mathbf{Z}} \left\{
    -  \frac{1}{\theta}\left(
      		e^{-\theta \left(\psi(t,x+\xi(z)) - \psi(t,x) \right)} -1
      	\right)
    	- \xi'(z) D\psi \right\}
      \nu(dz)
	\geq 0
																\nonumber
\end{equation}
\\

A locally bounded function $\Phi$ whose upper semicontinuous and lower semicontinuous envelopes are a viscosity subsolution and a viscosity supersolution of~\eqref{eq_JDRSAM_HJBPDE} is a viscosity solution of~\eqref{eq_JDRSAM_HJBPDE}.

\end{definition}

We would have similar definition for the viscosity supersolution, subsolution and solution of equation~\eqref{eq_JDRSAM_exptrans_HJBPDE}. Once again, the superjet and test function formulations are strictly equivalent (see Alvarez and Tourin~\cite{alto96} and Crandall, Ishii and Lions~\cite{crisli92}).
\\

\begin{remark}
An alternative, more classical, but also more restrictive definition of viscosity solution is as the continuous function which is both a supersolution and a subsolution of~\eqref{eq_JDRSAM_HJBPDE} (see Definition 5.1 in Barles~\cite{ba97}). The line of reasoning we will follow will make full use of the latitude afforded by our definition and we will have to wait until the comparison result is established in Section~\ref{sec_JDRSAM_comparison} to prove the continuity of the viscosity solution.
\end{remark}

\subsection{Characterization of the Value Function as a Viscosity Solution}
To show that the value function is a (discontinuous) viscosity solution of the associated RS HJB PIDE~\eqref{eq_JDRSAM_HJBPDE}, we follow an argument by Touzi~\cite{to02} which enables us to make a greater use of control theory in the derivation of the proof.
\\

%%%%%%%%%%%%%%%%%%%%%%%%%%%%%%%%%%%%%%%%%%%%%%%%%%%%%%%%%%%%%%%%%%%%%%%%%%%%%%
%   THEOREM
%%%%%%%%%%%%%%%%%%%%%%%%%%%%%%%%%%%%%%%%%%%%%%%%%%%%%%%%%%%%%%%%%%%%%%%%%%%%%%

\begin{theorem}\label{theo_JDRSAM_viscositysol}
$\Phi$ is a (discontinuous) viscosity solution of the RS HJB PIDE~\eqref{eq_JDRSAM_HJBPDE} on $[0,T] \times \mathbb{R}^n$, subject to terminal condition~\eqref{eq_JDRSAM_HJBPDE_termcond}.
\end{theorem}

\begin{proof}

\textbf{Outline - }
This proof can be decomposed in five steps. First, we define $\tilde{\Phi}$ as a log transformation of $\Phi$. In the next three steps, we prove that $\tilde{\Phi}$ is a viscosity solution of the exponentially transformed RS HJB PIDE by showing that it is 1). a viscosity subsolution, 2). a viscosity supersolution and hence 3). a viscosity solution. Finally, applying a change of variable result, such as Proposition 2.2 in~\cite{to02}, we conclude that $\Phi$ is a viscosity solution of the RS HJB PIDE~\eqref{eq_JDRSAM_HJBPDE}
\\

\textbf{Step 1: Exponential Transformation}\\

In order to prove that the value function $\Phi$ is a (discontinuous) viscosity solution of~\eqref{eq_JDRSAM_HJBPDE}, we will start by proving that the exponentially transformed value function $\tilde{\Phi}$ is a (discontinuous) viscosity solution of~\eqref{eq_JDRSAM_exptrans_HJBPDE}.
\\

\textbf{Step 2: Viscosity Subsolution}\\

Let $(t_0,x_0) \in Q := [0,t]\times\mathbb{R}^n$ and $u \in C^{1,2}(Q)$ satisfy
\begin{equation}
    0 = (\tilde{\Phi}^* - u)(t_0,x_0) = \max_{(t,x) \in Q} (\tilde{\Phi}^*(t,x) - u(t,x))
\end{equation}
and hence
\begin{equation}\label{eq_JDRSAM_theoviscositysol_step2_ineq1}
    \tilde{\Phi} \leq \tilde{\Phi}^* \leq u
\end{equation}
on $Q$
\\

Let $(t_k,x_k)$ be a sequence in $Q$ such that
\begin{equation}
    \lim_{k \to \infty} (t_k,x_k) = (t_0,x_0)
                \nonumber
\end{equation}
\begin{equation}
    \lim_{k \to \infty} \tilde{\Phi}(t_k,x_k) = \tilde{\Phi}^*(t_0,x_0)
                \nonumber
\end{equation}
and define the sequence $\left\{\xi\right\}_k$ as $\xi_k :=
\tilde{\Phi}(t_k,x_k) - u(t_k,x_k)$. Since $u$ is of class $C^{1,2}$, $\lim_{k \to \infty} \xi_k
= 0$.
\\

Fix $h \in \mathcal{J}$ and consider a constant control $\hat{h} = h$. Denote
by $X^k$ the state process with initial data $X_{t_k}^{k} = x_k$
and, for $k>0$, define the stopping time
\begin{equation}
    \tau_k := \inf\left\{
        s > t_k : (s-t_k,X_s^k-x_k) \notin
        [0,\delta_k)
        \times \alpha \mathscr{B}_n \right\}
                                    \nonumber
\end{equation}
for a given constant $\alpha > 0$ and where $\mathscr{B}_n$ is the
unit ball in $\mathbb{R}^n$ and
\begin{equation}
    \delta_k := \sqrt{\xi_k}\left(1-\mathit{1}_{\left\{0\right\}}(\xi_k)\right)
        + k^{-1}\mathit{1}_{\left\{0\right\}}(\xi_k)
                            \nonumber
\end{equation}
From the definition of $\tau_k$, we see that $\lim_{k \to \infty}
\tau_k = t_0$.
\\

By the Dynamic Programming Principle,
\begin{eqnarray}
    \tilde{\Phi}(t_k,x_k)
            &\leq& \mathbf{E}_{t_k,x_k} \left[ \exp \left\{ \theta \int_{t_k}^{\tau_k}
                g(X_s,\hat{h}_s;\theta) ds \right\} \tilde{\Phi}(\tau_k, X_{\tau_k}^{k})\right]     \nonumber
\end{eqnarray}
where $\mathbf{E}_{t_k,x_k} \left[ \cdot \right]$ represents the expectation under the measure $\mathbb{P}$ given initial data $(t_k,x_k)$.
\\

By inequality~\eqref{eq_JDRSAM_theoviscositysol_step2_ineq1},
\begin{eqnarray}
    \tilde{\Phi}(t_k,x_k)
            &\leq& \mathbf{E}_{t_k,x_k} \left[ \exp \left\{ \theta \int_{t_k}^{\tau_k}
                g(X_s,\hat{h}_s) ds \right\} u(\tau_k, X_{\tau_k}^{k})\right]
                                                                                \nonumber
\end{eqnarray}
and hence by definition of $\xi_k$,
\begin{eqnarray}
    u(t_k,x_k) + \xi_k
            &\leq& \mathbf{E}_{t_k,x_k} \left[ \exp \left\{ \theta \int_{t_k}^{\tau_k}
                g(X_s,\hat{h}_s) ds \right\} u(\tau_k, X_{\tau_k}^{k})\right]
                                                                                \nonumber
\end{eqnarray}
i.e.
\begin{eqnarray}
    \xi_k
    &\leq& \mathbf{E}_{t_k,x_k} \left[ \exp \left\{ \theta \int_{t_k}^{\tau_k}
            g(X_s,\hat{h}_s) ds \right\} u(\tau_k, X_{\tau_k}^{k})\right] - u(t_k,x_k)
                                                                                \nonumber
\end{eqnarray}

Define $Z(t_k) = \theta \int_{t_k}^{\tau_k}
g(X_s,\hat{h}_s) ds$, then
\begin{equation}
    d\left(e^{Z_s}\right) := \theta g(X_s,\hat{h}_s)e^{Z_s}ds
                                                                    \nonumber
\end{equation}
Also, by It\^o,
\begin{eqnarray}
    du_s
    &=& \left\{
        \frac{\partial u}{\partial s} + \mathcal{L}u \right\} ds
    +   Du'\Lambda(s) dW_{s}
                                                                    \nonumber\\
    &&
    +   \int_{\mathbf{Z}}\left\{
            u\left(s,X(s^-)+\xi(z)\right) - u\left(s,X(s^-)\right)
        \right\}\tilde{N}_{\textbf{p}}(ds,dz)
                                                                    \nonumber
\end{eqnarray}
for $s \in \left[t_k, \tau_k \right]$ and where the generator $\mathcal{L}$ of the state process $X(t)$ is defined as
\begin{eqnarray}\label{eq_JDRSAM_generator_X}
   \mathcal{L}u(t,x)
        &:=&	f(t,x,h;\theta)'Du
               + \frac{1}{2}\textrm{tr}\left(\Lambda\Lambda'(t,X)D^2 u \right)
\end{eqnarray}
\\

By the It\^o product rule, and since $dZ_s \cdot u_s = 0$, we get
\begin{equation}
    d\left(u_s e^{Z_s}\right) = u_s d\left(e^{Z_s}\right) + e^{Z_s} du_s
                                        \nonumber
\end{equation}
and hence for $t \in [t_k, \tau_k]$
\begin{eqnarray}
    u(t,X_t^k)e^{Z_t} &=&
        u(t_k,x_k)e^{Z_{t_k}}
        + \theta\int_{t_k}^{t} u(s,X_s^k)g(X_s^k,\hat{h}_s)e^{Z_s}ds
                                                    \nonumber\\
    &&  + \int_{t_k}^{t}
            \left( \frac{\partial u}{\partial s}(s,X_s^k)
            + \mathcal{L}u(s,X_s^k) e^{Z_s} \right)ds
        + \int_{t_k}^{t} Du'\Lambda(s) dW_{s}
                                                    \nonumber\\
    &&  + \int_{t_k}^{t}\int_{\mathbf{Z}}\left\{
            u\left(t,X^k(s^-)+\xi(z)\right) - u\left(t,X^k(s^-)\right)
        \right\}\tilde{N}_{\textbf{p}}(dt,dz)
                                                                    \nonumber
\end{eqnarray}

Noting that $u(t_k,x_k)e^{Z_{t_k}} = u(t_k,x_k)$ and taking the expectation with respect to the initial data $(t_k,x_k)$, we get
\begin{eqnarray}
    && \mathbf{E}_{t_k,x_k} \left[ u(t,X_t)e^{Z_t} \right]
                                                    \nonumber\\
    &=&  u(t_k,x_k)e^{Z_{t_k}}
        + \mathbf{E}_{t_k,x_k} \left[
          \int_{t_k}^{t}
                \left(\frac{\partial u}{\partial s}(s,X_s) + \mathcal{L}u(s,X_s)
                +\theta u(s,X_s)g(X_s,\hat{h}_s)\right)e^{Z_s}ds
          \right]
                                                    \nonumber
\end{eqnarray}

In particular, for $t = \tau_k$,
\begin{eqnarray}
    \xi_k
    &\leq& \mathbf{E}_{t_k,x_k} \left[ u(\tau_k,X_{\tau_k})e^{Z_{\tau_k}} \right]
             - u(t_k,x_k)e^{Z_{t_k}}
                                                    \nonumber\\
    &=&  + \mathbf{E}_{t_k,x_k} \left[
          \int_{t_k}^{\tau_k}
                \left(\frac{\partial u}{\partial s}(s,X_s) + \mathcal{L}u(s,X_s)
                +\theta u(s,X_s)g(X_s,\hat{h}_s)\right)e^{Z_s}ds
          \right]
                                                    \nonumber
\end{eqnarray}
and thus
\begin{eqnarray}
    \frac{\xi_k}{\delta_k}
    &\leq& \frac{1}{\delta_k}\left(
             \mathbf{E}_{t_k,x_k,} \left[ u(\tau_k,X_{\tau_k})e^{Z_{\tau_k}} \right]
             - u(t_k,x_k)e^{Z_{t_k}}\right)
                                                    \nonumber\\
    &=&  \frac{1}{\delta_k}\left(
          \mathbf{E}_{t_k,x_k} \left[
          \int_{t_k}^{\tau_k}
                \left(\frac{\partial u}{\partial s}(s,X_s) + \mathcal{L}u(s,X_s)
                +\theta u(s,X_s)g(X_s,\hat{h}_s)\right)e^{Z_s}ds
          \right] \right)
                                                    \nonumber
\end{eqnarray}

As $k \to \infty$, $t_k \to t_0$, $\tau_k \to t_0$, $\frac{\xi_k}{\delta_k} \to 0$ and
\begin{eqnarray}
    && \frac{1}{\delta_k}\left(
          \mathbf{E}_{t_k,x_k} \left[
          \int_{t_k}^{t}
                \left(\frac{\partial u}{\partial s}(s,X_s) + \mathcal{L}u(s,X_s)
                +\theta u(s,X_s)g(X_s,\hat{h}_s)\right)e^{Z_s}ds
          \right] \right)
                                                    \nonumber\\
    &\to&
        \frac{\partial u}{\partial s}(s,X_s) + \mathcal{L}u(s,X_s)
        +\theta u(s,X_s)g(X_s,\hat{h}_s)
                                                    \nonumber
\end{eqnarray}
a.s. by the Bounded Convergence Theorem, since the random variable
\begin{equation}
    \frac{1}{\delta_k} \int_{t_k}^{t} \left(
        \frac{\partial u}{\partial s}(s,X_s)
        + \mathcal{L}u(s,X_s)
        +\theta u(s,X_s)g(X_s,\hat{h}_s)
    \right)e^{Z_s}ds
                                                    \nonumber
\end{equation}
is bounded for large enough $k$.
\\

Hence, we conclude that since $\hat{h}_s$ is arbitrary,
\begin{equation}
    \frac{\partial u}{\partial s}(s,X_s) + \mathcal{L}u(s,X_s)
    +\theta u(s,X_s)g(X_s,\hat{h}_s)
    \geq 0
                                                    \nonumber
\end{equation}
i.e.
\begin{equation}
    -\frac{\partial u}{\partial s}(s,X_s) - \mathcal{L}u(s,X_s)
    -\theta u(s,X_s)g(X_s,\hat{h}_s)
    \leq 0
                                                    \nonumber
\end{equation}

This argument proves that $\tilde{\Phi}$ is a (discontinuous) viscosity subsolution of the PDE~\eqref{eq_JDRSAM_exptrans_HJBPDE} on $[0,t) \times \mathbb{R}^n$ subject to terminal condition $\tilde{\Phi}(T, x) = e^{g_0(x;T)}$.
\\

\textbf{Step 3: Viscosity Supersolution}\\

This step in the proof is a slight adaptation of the proof for classical control problems in Touzi~\cite{to02}. Let $(t_0,x_0) \in Q$ and $u \in C^{1,2}(Q)$ satisfy
\begin{equation}\label{eq_JDRSAM_theoviscositysol_step3_ineq1}
    0 = (\tilde{\Phi}_* - u)(t_0,x_0) < (\tilde{\Phi}_* - u)(t,x) \textrm{ for } Q\backslash{(t_0,x_0)}
\end{equation}

We intend to prove that at $(t_0,x_0)$
\begin{equation}
   \frac{\partial u}{\partial t}(t,x)
    + \inf_{h \in \mathcal{H}}\left\{
        \mathcal{L}^{h} u(t,x) - \theta g(x,h)
    \right\}
    \leq 0
                            \nonumber
\end{equation}
by contradiction. Thus, assume that
\begin{equation}\label{eq_JDRSAM_theoviscositysol_step3_contradictionPDE}
   \frac{\partial u}{\partial t}(t,x)
    + \inf_{h \in \mathcal{H}}\left\{
        \mathcal{L}^{h} u(t,x) - \theta g(x,h)
    \right\}
    > 0
\end{equation}
at $(t_0,x_0)$.

Since $\mathcal{L}^h u$ is continuous, there exists an open neighbourhood $\mathcal{N}_{\delta}$ of $(t_0,x_0)$ defined for $\delta > 0$ as
\begin{equation}
    \mathcal{N}_{\delta} := \left\{(t,x): (t-t_0,x-x_0) \in (-\delta,\delta) \times \delta \mathscr{B}_n, \textrm{ and~\eqref{eq_JDRSAM_theoviscositysol_step3_contradictionPDE} holds}\right\}
\end{equation}
\\

Note that by~\eqref{eq_JDRSAM_theoviscositysol_step3_ineq1} and since $\tilde{\Phi} > \tilde{\Phi}_{*} > u $,
\begin{equation}
    \min_{Q \backslash \mathcal{N}_{\delta}}\left(\tilde{\Phi} - u \right) > 0
                \nonumber
\end{equation}
\\

For $\rho >0$, consider the set $J^{\rho}$ of $\rho$-optimal controls $h^{\rho}$ satisfying
\begin{equation}\label{eq_JDRSAM_HJBderivation_step3_def_epsilonoptimal}
    \tilde{I}(t_0,x_0,h^{\rho}) \leq \tilde{\Phi}(t_0,x_0) + \rho
\end{equation}
\\

Also, let $\epsilon > 0$, $\epsilon \leq \gamma$ be such that
\begin{equation}\label{eq_JDRSAM_visc_HJBderivation_step3_epsilon}
   \min_{Q \backslash \mathcal{N}_{\delta}}
	\left(\tilde{\Phi} - u \right)
	\geq 3\epsilon e^{-\delta \theta M_{\delta}}
    > 0
\end{equation}
where $M_{\delta}$ is defined as
\begin{equation}
    M_{\delta} := \max_{(t,x) \in \mathcal{N}_{\delta}^J, h \in \mathcal{J}^{\rho}}\left(-g(x,h),0\right)
                                        \nonumber
\end{equation}
for
\begin{equation}
    \mathcal{N}_{\delta}^J := \left\{(t,x): (t-t_0,x-x_0) \in (-\delta,\delta) \times (\zeta+\delta) \mathscr{B}_n \right\}
\end{equation}
and
\begin{equation}
    \zeta := \max_{z \in \mathbb{Z}} \| \xi(z)\|
                                        \nonumber
\end{equation}
Note that $\zeta < \infty$ by boundedness of $\xi(z)$ and thus $M_{\delta} < \infty$.
\\

Now let $(t_k,x_k)$ be a sequence in $\mathcal{N}_{\delta}$ such that
\begin{equation}
    \lim_{k \to \infty} (t_k,x_k) = (t_0,x_0)
                \nonumber
\end{equation}
and
\begin{equation}
    \lim_{k \to \infty} \tilde{\Phi}(t_k,x_k) = \tilde{\Phi}_*(t_0,x_0)
                \nonumber
\end{equation}
Since $(\tilde{\Phi}-u)(t_k,x_k) \to 0$, we can assume that the sequence $(t_k,x_k)$ satisfies
\begin{equation}\label{eq_JDRSAM_theoviscositysol_step3_ineq2}
    \lvert (\tilde{\Phi}-u)(t_k,x_k) \rvert
    \leq  \epsilon,
    \qquad \textrm{for } k \geq 1
\end{equation}
for $\epsilon$ defined by~\eqref{eq_JDRSAM_visc_HJBderivation_step3_epsilon}
\\

Consider the $\epsilon$-optimal control $h_k^{\epsilon}$, denote by $\tilde{X}_k^\epsilon$ the controlled process defined by the control process $h_k^{\epsilon}$ and introduce the stopping time
\begin{equation}
    \tau_k := \inf\left\{s>\tau_k : (s,\tilde{X}_k^\epsilon(s)) \notin \mathcal{N}_{\delta} \right\}
                            \nonumber
\end{equation}
Note that since we assumed that $-\infty \leq \xi_{i}^{\textrm{min}} \leq \xi_{i} \leq \xi_{i}^{\textrm{max}} < \infty$ for $i = 1, \ldots, n$ and since $\nu$ is assumed to be bounded then $X(\tau)$ is also finite and in particular,
\begin{equation}\label{eq_JDRSAM_theoviscositysol_step3_ineq3}
    (\tilde{\Phi}-u)(\tau_k,\tilde{X}_k^\epsilon(\tau_k))
    \geq (\tilde{\Phi}_{*} - u)(\tau_k,\tilde{X}_k^\epsilon(\tau_k))
    \geq 3\epsilon e^{-\delta \theta M_{\delta}}
\end{equation}
\\

Choose $\mathcal{N}_{\delta}^{J}$ so that $(\tau,\tilde{X}^\epsilon(\tau)) \in \mathcal{N}_{\delta}^{J}$. In particular, since $X^\epsilon(\tau)$ is finite then $\mathcal{N}_{\delta}^{J}$  can be defined to be a strict subset of $Q$ and we can effectively use the local boundedness of $g$ to establish $M_{\delta}$.
\\

Let $Z(t_k) = \theta \int_{t_k}^{\bar{\tau}_k} g(\tilde{X}_s^{\epsilon},h_s^{\epsilon}) ds$, since $\tilde{\Phi} \geq \tilde{\Phi}_{*}$ and by~\eqref{eq_JDRSAM_theoviscositysol_step3_ineq2} and~\eqref{eq_JDRSAM_theoviscositysol_step3_ineq3},
\begin{eqnarray}
    &&  \tilde{\Phi}(\tau_k,\tilde{X}_k^\epsilon(\tau_k))e^{Z(\tau_k)}
        -\tilde{\Phi}(t_k,x_k)e^{Z(t_k)}
                                        \nonumber\\
    &\geq&
        u(\tau_k,\tilde{X}_k^\epsilon(\tau_k))e^{Z(\tau_k)}
        -\tilde{\Phi}(t_k,x_k)e^{Z(t_k)}
        + 3\epsilon e^{-\delta \theta M_{\delta}}e^{Z(\tau_k)}
        - \epsilon
                                \nonumber\\
    &\geq&
        \int_{t_k}^{\tau_k} d\left(u(s,\tilde{X}_k^\epsilon(s))e^{Z_s}\right)
        + 2\epsilon
                                \nonumber
\end{eqnarray}
i.e.
\begin{eqnarray}
   \tilde{\Phi}(t_k,x_k)
    &\leq&
        \tilde{\Phi}(\tau_k,\tilde{X}_k^\epsilon(\tau_k))e^{Z(\tau_k)}
        - \int_{t_k}^{\tau_k} d\left(u(s,\tilde{X}_k^\epsilon(s))e^{Z_s}\right)
        - 2\epsilon
                            \nonumber
\end{eqnarray}

Taking expectation with respect to the initial data $(t_k,x_k)$,
\begin{eqnarray}
   \tilde{\Phi}(t_k,x_k)
    &\leq&
        \mathbf{E}_{t_k,x_k}\left[
            \tilde{\Phi}(\tau_k,\tilde{X}_k^\epsilon(\tau_k))e^{Z(\tau_k)}
            - \int_{t_k}^{\tau_k} d\left(u(s,\tilde{X}_k^\epsilon(s))e^{Z_s}\right)
        \right]
        - 2\epsilon
                            \nonumber
\end{eqnarray}

Note that by the It\^o product rule,
\begin{eqnarray}
    &&   d\left(u(s,\tilde{X}_k^\epsilon(s))e^{Z_s}\right)
                                    \nonumber\\
    &=& u_s d\left(e^{Z_s}\right) + e^{Z_s} du_s
                                    \nonumber\\
    &=& \frac{\partial u}{\partial t}(t,x)+ \mathcal{L}^{h} u(t,x)
            + \theta g(x,h)
                                    \nonumber
\end{eqnarray}
Since we assumed that
\begin{equation}
   -\frac{\partial u}{\partial t}(t,x)- \mathcal{L}^{h} u(t,x)
    -\theta g(x,h) < 0
                            \nonumber
\end{equation}
then
\begin{equation}
    -\int_{t_k}^{\tau_k} d\left(u(s,\tilde{X}_k^\epsilon(s))e^{z_s}\right)
     < 0
                            \nonumber
\end{equation}
and therefore
\begin{eqnarray}
    \tilde{\Phi}(t_k,x_k)
    &\leq&
        \mathbf{E}_{t_k,x_k}\left[
            \tilde{\Phi}(\tau_k,\tilde{X}_k^\epsilon(\tau_k))e^{Z(\tau_k)}
            - \int_{t_k}^{\tau_k} d\left(u(s,\tilde{X}_k^\epsilon(s))e^{Z_s}\right)
        \right]
        - 2\epsilon
                            \nonumber\\
    &\leq&  -2\epsilon + \mathbf{E}
        \left[ \exp \left\{ \theta \int_{t_k}^{\tau_k} g(X_s,h_k^{\epsilon}(s))
        ds\right\}\tilde{\Phi}(\tau_k,\tilde{X}_k^\epsilon(\tau_k))\right]
                            \nonumber\\
    &\leq&  -2\epsilon + \tilde{I}(t_k,x_k,h_k^{\epsilon})
                            \nonumber\\
    &\leq&  \tilde{\Phi}(t_k,x_k) - \epsilon
                            \nonumber
\end{eqnarray}
where the third inequality follows from the Dynamic Programming Principle and the last inequality follows from the definition of $\epsilon$-optimal controls (see equation~\eqref{eq_JDRSAM_HJBderivation_step3_def_epsilonoptimal}).
\\

Hence, equation~\eqref{eq_JDRSAM_theoviscositysol_step3_contradictionPDE},
\begin{equation}
   \frac{\partial u}{\partial t}(t,x)
    + \inf_{h \in \mathcal{H}}\left\{
        \mathcal{L}^{h} u(t,x) - \theta g(x,h)
    \right\}
    > 0
                                            \nonumber
\end{equation}
is false and we have shown that
\begin{equation}
   \frac{\partial u}{\partial t}(t,x)
    + \inf_{h \in \mathcal{H}}\left\{
        \mathcal{L}^{h} u(t,x) - \theta g(x,h)
    \right\}
    \leq 0
                                            \nonumber
\end{equation}
\\

This argument therefore proves that $\tilde{\Phi}$ is a (discontinuous) viscosity supersolution of the PDE~\eqref{eq_JDRSAM_exptrans_HJBPDE} on $[0,t) \times \mathbb{R}^n$ subject to terminal condition $\tilde{\Phi}(T, x) = e^{g_0(x;T)}$.
\\

\textbf{Step 4: Viscosity Solution}\\

Since $\tilde{\Phi}$ is both a (discontinuous) viscosity subsolution and a supersolution of~\eqref{eq_JDRSAM_exptrans_HJBPDE}, it is a (discontinuous) viscosity.
\\

\textbf{Step 5: Conclusion}\\

Since by assumption $\Phi$ is locally bounded, so is $\tilde{\Phi}$. In addition, $\varphi(x) = e^{-\theta x}$ is of class $C_1(\mathbb{R})$. Also we note that $\frac{d\varphi}{dx} < 0$. By the change of variable property (see for example Proposition 2.2 in Touzi~\cite{to02}), we see that
\begin{enumerate}[1.]
\item since $\tilde{\Phi}$ is a (discontinuous) viscosity subsolution of~\eqref{eq_JDRSAM_exptrans_HJBPDE}, $\Phi = \varphi^{-1} \circ \tilde{\Phi}$ is a (discontinuous) viscosity supersolution of~\eqref{eq_JDRSAM_HJBPDE};
\item since $\tilde{\Phi}$ is a (discontinuous) viscosity supersolution of~\eqref{eq_JDRSAM_exptrans_HJBPDE}, $\Phi = \varphi^{-1} \circ \tilde{\Phi}$ is a (discontinuous) viscosity subsolution of~\eqref{eq_JDRSAM_HJBPDE}.
\end{enumerate}
and therefore $\Phi$ is a (discontinuous) viscosity solution of~\eqref{eq_JDRSAM_HJBPDE} on $[0,t) \times \mathbb{R}^n$ subject to terminal condition $\tilde{\Phi}(T, x) = e^{g_0(x;T)}$.

\end{proof}

We also note the following corollary:

%%%%%%%%%%%%%%%%%%%%%%%%%%%%%%%%%%%%%%%%%%%%%%%%%%%%%%%%%%%%%%%%%%%%%%%%%%%%%%
%   Corollary
%%%%%%%%%%%%%%%%%%%%%%%%%%%%%%%%%%%%%%%%%%%%%%%%%%%%%%%%%%%%%%%%%%%%%%%%%%%%%%

\begin{corollary}\label{coro_JDRSAM_viscositysol}
\begin{enumerate}[(i).]
    \item{} $\Phi^*$ is a upper semicontinuous viscosity subsolution, and;
    \item{} $\Phi_*$ is a lower semicontinuous viscosity supersolution
\end{enumerate}
of the RS HJB PIDE~\eqref{eq_JDRSAM_HJBPDE} on $[0,T] \times \mathbb{R}^n$, subject to terminal condition~\eqref{eq_JDRSAM_HJBPDE_termcond}.
\end{corollary}

As a result of this corollary, we note that $\Phi^*$, $\Phi_*$ and $\Phi$ are respectively a viscosity subsolution, supersolution, and solution in the sense of Definitions~\ref{def_JDRSAM_viscositysol_semijet} and~\ref{def_JDRSAM_viscositysol_testfunc}.
\\

%%%%%%%%%%%%%%%%%%%%%%%%%%%%%%%%%%%%%%%%%%%%%%%%%%%%%%%%%%%%%%%%%%%%%%%%%%%%%%%%%%%
%
%   NEXT SECTION
%
%%%%%%%%%%%%%%%%%%%%%%%%%%%%%%%%%%%%%%%%%%%%%%%%%%%%%%%%%%%%%%%%%%%%%%%%%%%%%%%%%%%

%%%
\section{Comparison Result}\label{sec_JDRSAM_comparison}
Once we have characterized the class of viscosity solutions
associated with a given problem, the next task is to prove that
the problem actually admits a unique viscosity solution by establishing a comparison theorem. Comparison theorems are the cornerstone of the application of viscosity theory. Their main use is to prove uniqueness, and in our case continuity, of the viscosity solution. Although a set of, by now fairly standard, techniques can be applied in the proof, the comparison theorem \textit{per se} is generally customized to address both the specificities of the PDE and the requirements of the general problem.
\\

We face three main difficulties in establishing a comparison result for our risk-sensitive control problem. The first obstacle is the behaviour of the value function $\Phi$ at infinity. In the pure diffusion case or LEQR case solved by Kuroda and Nagai~\cite{kuna02}, the value function is quadratic in the state and is therefore not bounded for $x \in \mathbb{R}^n$. Consequently, there is no reason to expect the solution to the integro-differential RS HJB PIDE~\eqref{eq_JDRSAM_HJBPDE} to be bounded. The second hurdle is the presence of an extra non-linearity: the quadratic growth term $\left(D\Phi\right)' \Lambda\Lambda' D\Phi$. This extra non-linearity could, in particular, increase the complexity of the derivation of a comparison result for an unbounded value function. Before dealing with the asymptotic growth condition we will therefore need to address this non-linear term. The traditional solution, an exponential change of variable such as the one proposed by Duffie and Lions~\cite{duli92}, is equivalent to the log transformation we used to derive the RS HJB PIDE and again to prove that the value function is a viscosity solution of the RS HJB PIDE. However, the drawback of this method is that, by creating a new zeroth order term equal to the solution multiplied by the cost function $g$, it imposes a severe restriction on $g$ for the PDE to satisfy the monotonicity property required to talk about viscosity solutions. The final difficulty lies in the presence of the jump term and of the compensator $\nu$. If we assume that the measure is finite, this can be addressed following the general argument proposed by Alvarez and Tourin~\cite{alto96} and Amadori~\cite{am00}. 
\\

To address these difficulties, we will need to adopt a slightly different strategy from the classical argument used to proof comparison results as set out in Crandall, Ishii and Lions~\cite{crisli92}. In particular, we will exploit the properties of the exponentially transformed value function $\tilde{\Phi}$ resulting from Assumption~\ref{as_JDRSAM_ODEs_alpha_beta} and alternate between the log transformed RS HJB PIDE and the quadratic growth RS HJB PIDE~\eqref{eq_JDRSAM_HJBPDE} through the proof.
\\

%%%%%%%%%%%%%%%%%%%%%%%%%%%%%%%%%%%%%%%%%%%%%%%%%%%%%%%%%%%%%%%%%%%%%%%%%%%%%%
%   THEOREM
%%%%%%%%%%%%%%%%%%%%%%%%%%%%%%%%%%%%%%%%%%%%%%%%%%%%%%%%%%%%%%%%%%%%%%%%%%%%%%

\begin{theorem}[Comparison Result on an Unbounded State Space]\label{theo_JDRSAM_comparison_unbounded}
    Let $\tilde{u} = e^{-\theta v} \in USC([0,T]\times\mathbb{R}^{n})$ be a bounded from above viscosity subsolution of~\eqref{eq_JDRSAM_HJBPDE} and $\tilde{v} = e^{-\theta u} \in LSC([0,T]\times\mathbb{R}^{n})$ be a bounded from below viscosity supersolution of~\eqref{eq_JDRSAM_HJBPDE}. If the measure $\nu$ is bounded and Assumption~\ref{as_JDRSAM_ODEs_alpha_beta} holds then
\begin{equation}
    u \leq v
    \quad \textrm{on } [0,T] \times \mathbb{R}^n
                                \nonumber
\end{equation}
\end{theorem}

\begin{proof}
\textbf{Outline - }
This proof can be decomposed in seven steps. In the first step, we perform the usual exponential transformation to rewrite the problem for the value function $\Phi$ into a problem for the value function $\tilde{\Phi}$. The rest of the proof is done by contradiction. In step 2, we state the assumption we are planning to disprove. The properties of the value function $\tilde{\Phi}$ related to Assumption~\ref{as_JDRSAM_ODEs_alpha_beta} are used in Step 3 to deduce that it is enough to prove the comparison result for $\Phi$ on a bounded state space to reach our conclusion. We then double variables in step 4 before finding moduli of continuity for the diffusion and the jump components respectively in steps 5 and 6. Finally, we reach a contradiction in step 7 and conclude the proof.
\\

\textbf{Step 1: Exponential Transformation}\\

Let $u \in USC([0,T]\times\mathbb{R}^{n})$ be a viscosity subsolution of~\eqref{eq_JDRSAM_HJBPDE} and $v \in LSC([0,T]\times\mathbb{R}^{n})$ be a viscosity supersolution of~\eqref{eq_JDRSAM_HJBPDE}. Define:
\begin{eqnarray}
    \tilde{u} := e^{-\theta v}
                                                            \nonumber\\
    \tilde{v} := e^{-\theta u}
                                                            \nonumber
\end{eqnarray}
By the change of variable property (see for example Proposition 2.2 in Touzi~\cite{to02}), $\tilde{u}$ and $\tilde{v}$ are respectively a viscosity subsolution and a viscosity supersolution of the RS HJB PIDE~\eqref{eq_JDRSAM_exptrans_HJBPDE} for the exponentially transformed value function $\tilde{\Phi}$.
\\

Thus, to prove that
\begin{equation}
    u \leq v
    \quad \textrm{on } [0,T] \times\mathbb{R}^n
                                \nonumber
\end{equation}
it is sufficient to prove that
\begin{equation}
    \tilde{u} \leq \tilde{v}
    \quad \textrm{on } [0,T] \times\mathbb{R}^n
                                \nonumber
\end{equation}
\\

\textbf{Step 2: Setting the Problem}\\
As is usual in the derivation of comparison results, we argue by contradiction and assume that
\begin{eqnarray}\label{as_JDRSAM_comparison_theo_contradiction}
    \sup_{(t,x) \in [0,T] \times\mathbb{R}^n} \left[ \tilde{u}(t,x) - \tilde{v}(t,x) \right] >0
\end{eqnarray}
\\

\textbf{Step 3: Taking the Behaviour of the Value Function into Consideration}\\

The assertion of this theorem is that the comparison result holds in the class of functions satisfying Assumption~\ref{as_JDRSAM_ODEs_alpha_beta}. As a result Proposition~\ref{prop_JDRSAM_behaviour_tildePhi} holds and we can concentrate our analysis on subsolutions and supersolutions sharing the same growth properties as the exponentially transformed value function $\tilde{\Phi}$. By Propositions~\ref{prop_JDRSAM_behaviour_tildePhi} and~\ref{prop_JDRSAM_tildePhi_bounded},
\begin{equation}
    0 < \tilde{u}(t,x) \leq e^{\alpha^k(t)+\beta^{k'}(t)x}
    \quad \forall (t,x) \in [0,T]\times \mathbb{R}^n
                                        \nonumber
\end{equation}
\begin{equation}
    0 < \tilde{v}(t,x) \leq e^{\alpha^k(t)+\beta^{k'}(t)x}
    \quad \forall (t,x) \in [0,T]\times \mathbb{R}^n
                                        \nonumber
\end{equation}
and
\begin{equation}\label{eq_JDRSAM_CompRes_Step3_assymptotic}
    \lim_{\lvert x \rvert \to \infty} \tilde{u}(t,x)
    = \lim_{\lvert x \rvert \to \infty} \tilde{v}(t,x)
    = 0
    \; \forall t \in [0,T]
\end{equation}
for $k=1,\ldots,2n$ where $\alpha^k$ and $\beta^k$ are the functions given in Assumption~\ref{as_JDRSAM_ODEs_alpha_beta}. Since~\eqref{eq_JDRSAM_CompRes_Step3_assymptotic} holds at an exponential rate, then by Assumption~\eqref{as_JDRSAM_comparison_theo_contradiction} there exists $R>0$, such that
\begin{eqnarray}
    \sup_{(t,x)\in[0,T] \times\mathbb{R}^n} \left[ \tilde{u}(t,x) - \tilde{v}(t,x) \right]
    = \sup_{(t,x)\in[0,T] \times\mathcal{B}_{R}} \left[ \tilde{u}(t,x) - \tilde{v}(t,x) \right]
                                        \nonumber
\end{eqnarray}
Hence, it is enough to show a contradiction with respect to the hypothesis
\begin{eqnarray}\label{as_JDRSAM_comparison_theo_contradiction_Q}
    \sup_{(t,x) \in Q} \left[ \tilde{u}(t,x) - \tilde{v}(t,x) \right] >0
\end{eqnarray}
established on the set $Q := [0,T]\times \mathcal{B}_R$. Before proceeding to the next step, we will restate assumption~\eqref{as_JDRSAM_comparison_theo_contradiction_Q} now needs to be restated in terms of $u$ and $v$ as
\begin{eqnarray}\label{as_JDRSAM_comparison_theo_contradiction_Phi}
    \sup_{(t,x) \in Q} \left[ u(t,x) - v(t,x) \right] >0
\end{eqnarray}
\\

\textbf{Step 4:  Doubling of Variables on the Set $Q$}\\

Let $\eta > 0$ be such that
\begin{eqnarray}
    N := \sup_{(t,x) \in Q} \left[ u(t,x) - v(t,x) -\varphi(t) \right] >0
                                            \nonumber
\end{eqnarray}
where $\varphi(t) := \frac{\eta}{t}$.
\\

We will now double variables, a technique commonly used in viscosity
solutions literature (see e.g. Crandall, Ishii and
Lions~\cite{crisli92}). Consider a global maximum point
$(t_{\epsilon},x_{\epsilon},y_{\epsilon}) \in
(0,T]\times \mathcal{B}_R\times \mathcal{B}_R =: Q_d$ of
\begin{eqnarray}
    u(t,x) - v(t,y)
    - \varphi(t)
    - \epsilon\lvert x-y \rvert^2
                                            \nonumber
\end{eqnarray}
and define
\begin{eqnarray}
    N_{\epsilon}
    := \sup_{(t,x,y) \in Q_d}\left[
          u(t,x)
        - v(t,y)
        - \varphi(t)
        - \epsilon \lvert x - y \rvert^2
    \right]>0
                                            \nonumber
\end{eqnarray}
Note that $N_{\epsilon} > 0$ for $\epsilon$ large enough.  Moreover, $N_{\epsilon} \geq N$ and $N_{\epsilon} \downarrow 0$ as $\epsilon \to \infty$.
\\

It is well established (see Lemma 3.1 and Proposition 3.7 in~\cite{crisli92}) that along
a subsequence
\begin{eqnarray}
    \lim_{\epsilon \to \infty} (t_{\epsilon},x_{\epsilon},y_{\epsilon})
    = (\hat{t},\hat{x},\hat{x})
                                            \nonumber
\end{eqnarray}
for some $(\hat{t},\hat{x}) \in [0,T]\times\mathbb{R}^n$ which
is a maximum point of
\begin{eqnarray}
    u(t,x) - v(t,x)
    -\varphi(t)
                                            \nonumber
\end{eqnarray}
Via the same argument, we also have
\begin{equation}
    \lim_{\epsilon \to \infty} \epsilon\lvert x_{\epsilon}-y_{\epsilon} \rvert^2 = 0
                                                            \nonumber
\end{equation}
as well as
\begin{eqnarray}
    \lim_{\epsilon \to \infty} u(t_{\epsilon},x_{\epsilon}) = u(\hat{t},\hat{x})
                                            \nonumber
\end{eqnarray}
and
\begin{eqnarray}
    \lim_{\epsilon \to \infty} v(t_{\epsilon},x_{\epsilon}) = v(\hat{t},\hat{x})
                                            \nonumber
\end{eqnarray}
In addition, we note that
\begin{eqnarray}
    \lim_{\epsilon \to \infty} N_{\epsilon} = N
                                            \nonumber
\end{eqnarray}
\\

Applying Theorem 8.3 in Crandall, Ishii and Lions~\cite{crisli92} at $(t_{\epsilon},x_{\epsilon},y_{\epsilon})$, we see that there exists $a_{\epsilon},b_{\epsilon} \in \mathbb{R}$ and $A_{\epsilon},B_{\epsilon} \in \mathcal{S}_n$ such that
\begin{eqnarray}
    \left(a_{\epsilon}, \epsilon(x_{\epsilon}-y_{\epsilon}), A_{\epsilon} \right) \in \overline{\mathcal{P}}_{u}^{2,+}
                                            \nonumber\\
    \left(b_{\epsilon}, \epsilon(x_{\epsilon}-y_{\epsilon}), B_{\epsilon}\right) \in \overline{\mathcal{P}}_{v}^{2,-}
                                            \nonumber
\end{eqnarray}
\begin{equation}
    a_{\epsilon} - b_{\epsilon} = \varphi'(t_{\epsilon})
                                            \nonumber
\end{equation}
and
\begin{equation}
    -3\epsilon
    \left[\begin{array}{cc}
        I   &   0   \\
        0   &   I
    \end{array}\right]
    \leq
    \left[\begin{array}{cc}
        A_{\epsilon}    &   0                       \\
        0                   &   -B_{\epsilon}
    \end{array}\right]
    \leq
    3\epsilon
    \left[\begin{array}{cc}
        I   &   -I  \\
        -I  &   I
    \end{array}\right]
                                            \nonumber
\end{equation}

Thus, we have for the subsolution $u$
\begin{eqnarray}
    &&
    - a_{\epsilon}
    + F(x_{\epsilon},\epsilon(x_{\epsilon}-y_{\epsilon}),A_{\epsilon})
                                                \nonumber\\
    &&
    +   \int_{\mathbf{Z}} \left\{
        \frac{1}{\theta}\left(
                e^{-\theta \left(u(t_{\epsilon},x_{\epsilon}+\xi(z)) - u(t_{\epsilon},x_{\epsilon}) \right)} - 1
            \right)
        + \epsilon\xi'(z) (x_{\epsilon}-y_{\epsilon}) \right\}
        \nu(dz)
                                                \nonumber\\
    &&
    \leq 0
                                                \nonumber
\end{eqnarray}
and for the supersolution $v$,
\begin{eqnarray}
    &&
    - b_{\epsilon}
    + F(y_{\epsilon},\epsilon(x_{\epsilon}-y_{\epsilon}),B_{\epsilon})
                                                \nonumber\\
    &&
    +   \int_{\mathbf{Z}} \left\{
        \frac{1}{\theta}\left(
                e^{-\theta \left(v(t_{\epsilon},y_{\epsilon}+\xi(z)) - v(t_{\epsilon},y_{\epsilon}) \right)} - 1
            \right)
        + \epsilon\xi'(z) (x_{\epsilon}-y_{\epsilon}) \right\}
        \nu(dz)
                                                \nonumber\\
    &&
    \geq 0
                                                \nonumber
\end{eqnarray}
\\

Subtracting these two inequalities,
\begin{eqnarray}\label{eq_JDRSAM_theo_comp_bd_ineq1_unbd}
    - \varphi'(t_{\epsilon})
    &=& b_{\epsilon} - a_{\epsilon}
                                                \nonumber\\
    &\leq&
        F(y_{\epsilon},\epsilon(x_{\epsilon}-y_{\epsilon}),B_{\epsilon})
    -   F(x_{\epsilon},\epsilon(x_{\epsilon}-y_{\epsilon}),A_{\epsilon})
                                    \nonumber\\
    &&
    +   \int_{\mathbf{Z}} \left\{
        \frac{1}{\theta}\left(
                e^{-\theta \left(v(t_{\epsilon},y_{\epsilon}+\xi(z)) - v(t_{\epsilon},y_{\epsilon}) \right)} - 1
            \right)
        + \epsilon\xi'(z) (x_{\epsilon}-y_{\epsilon}) \right\}
        \nu(dz)
                                    \nonumber\\
    &&
    -   \int_{\mathbf{Z}} \left\{
        \frac{1}{\theta}\left(
                e^{-\theta \left(u(t_{\epsilon},x_{\epsilon}+\xi(z)) - u(t_{\epsilon},x_{\epsilon}) \right)} - 1
            \right)
        + \epsilon\xi'(z) (x_{\epsilon}-y_{\epsilon}) \right\}
        \nu(dz)
                                    \nonumber\\
   &=&
        F(y_{\epsilon},\epsilon(x_{\epsilon}-y_{\epsilon}),B_{\epsilon})
    -   F(x_{\epsilon},\epsilon(x_{\epsilon}-y_{\epsilon}),A_{\epsilon})
                                    \nonumber\\
    &&
    +   \frac{1}{\theta}\int_{\mathbf{Z}} \left\{
                e^{-\theta \left(v(t_{\epsilon},y_{\epsilon}+\xi(z)) - v(t_{\epsilon},y_{\epsilon}) \right)}
        \right\}\nu(dz)
                                    \nonumber\\
    &&
    -   \frac{1}{\theta}\int_{\mathbf{Z}} \left\{
                e^{-\theta \left(u(t_{\epsilon},x_{\epsilon}+\xi(z)) - u(t_{\epsilon},x_{\epsilon}) \right)}
        \right\}\nu(dz)
\end{eqnarray}

\textbf{Step 5: Modulus of Continuity}\\
In this step, we focus on the (diffusion) operator $F$.
\begin{eqnarray}
    && F(y_{\epsilon},\epsilon(x_{\epsilon}-y_{\epsilon}),B_{\epsilon}) - F(x_{\epsilon},\epsilon(x-y),A_{\epsilon})
                                        \nonumber\\
    &=& \sup_{h \in \mathcal{J}}\left\{
        \epsilon f(t_{\epsilon},y_{\epsilon},h)'\left(x_{\epsilon}-y_{\epsilon} \right)
      + \frac{1}{2} \textrm{tr}\left(\Lambda \Lambda' B_{\epsilon}\right)
      - \frac{\theta}{2} \epsilon^2 \left(x_{\epsilon}-y_{\epsilon} \right)'\Lambda\Lambda'\left(x_{\epsilon}-y_{\epsilon} \right)
      - g(y_{\epsilon},h)
        \right\}
                                        \nonumber\\
    &&
        - \sup_{h \in \mathcal{J}}\left\{
              \epsilon f(t_{\epsilon},x_{\epsilon},h)'\left(x_{\epsilon}-y_{\epsilon} \right)
            + \frac{1}{2} \textrm{tr}\left(\Lambda \Lambda' A_{\epsilon}+ \delta I_n\right)
        \right.
                                        \nonumber\\
    &&
        \left.
            - \frac{\theta}{2} \epsilon^2 \left(x_{\epsilon}-y_{\epsilon} \right)'\Lambda\Lambda'\left(x_{\epsilon}-y_{\epsilon} \right)
            - g(x_{\epsilon},h)
        \right\}
                                        \nonumber\\
    &\leq&
          \frac{1}{2} \vert \textrm{tr} \left(\Lambda \Lambda' B_{\epsilon} -\Lambda \Lambda'A_{\epsilon}\right) \rvert
        + \sup_{h \in \mathcal{J}} \left\{
              \epsilon\lvert f(t_{\epsilon},y_{\epsilon},h) - f(t_{\epsilon},x_{\epsilon},h) \rvert \lvert (x_{\epsilon}-y_{\epsilon}) \rvert
        \right\}
                                        \nonumber\\
    &&  + \sup_{h \in \mathcal{J}} \left\{
        \lvert g(x_{\epsilon},h) - g(y_{\epsilon},h)\rvert \right\}
                                        \nonumber\\
    &\leq&
          \frac{1}{2} \vert \textrm{tr} \left(\Lambda \Lambda'A_{\epsilon} - \Lambda \Lambda' B_{\epsilon}\right)\rvert
        + \sup_{h \in \mathcal{J}} \left\{
              \epsilon\lvert f(t_{\epsilon},y_{\epsilon},h) - f(t_{\epsilon},x_{\epsilon},h) \rvert \lvert (x_{\epsilon}-y_{\epsilon}) \rvert
        \right\}
                                        \nonumber\\
    &&  + \sup_{h \in \mathcal{J}} \left\{
        \lvert g(x_{\epsilon},h) - g(y_{\epsilon},h)\rvert \right\}
                                        \nonumber
\end{eqnarray}
\\

Note that the functional $f$ defined in~\eqref{eq_JDRSAM_func_f} satisfies
\begin{equation}
    \left| f(t_{\epsilon},y_{\epsilon},h) - f(t_{\epsilon},x_{\epsilon},h) \right|
    \leq C_f \left| y_{\epsilon} - x_{\epsilon} \right|
                                                                \nonumber
\end{equation}
for some constant $C_f>0$. In addition,
\begin{eqnarray}
    &&\textrm{tr} \left(\Lambda \Lambda' A_{\epsilon} -\Lambda \Lambda'B_{\epsilon}\right)
                                                    \nonumber\\
    &=& \textrm{tr} \left(
            \left[\begin{array}{cc}
                \Lambda \Lambda'     &   \Lambda \Lambda'   \\
                \Lambda \Lambda'     &   \Lambda \Lambda'
            \end{array}\right]
            \left[\begin{array}{cc}
                A_{\epsilon}    &   0               \\
                0               &   -B_{\epsilon}
            \end{array}\right]
         \right)
                                                    \nonumber\\
    &\leq& 3\epsilon \textrm{ tr} \left(
            \left[\begin{array}{cc}
                \Lambda \Lambda'     &   \Lambda \Lambda'   \\
                \Lambda \Lambda'     &   \Lambda \Lambda'
            \end{array}\right]
            \left[\begin{array}{cc}
                I   &   -I  \\
                -I  &   I
            \end{array}\right]
         \right)
                                                    \nonumber\\
    &=& 0
                                                    \nonumber
\end{eqnarray}
Finally, by definition of $g$,
\begin{equation}
    \left| g(y_{\epsilon},h) - g(x_{\epsilon},h) \right|
    \leq C_g \left| y_{\epsilon} - x_{\epsilon} \right|
                                                                \nonumber
\end{equation}
for some constant $C_g>0$. Combining these estimates, we get
\begin{eqnarray}\label{eq_JDRSAM_theo_comp_modulus_continuity}
    && F(y_{\epsilon},\epsilon(x_{\epsilon}-y_{\epsilon}),B_{\epsilon}) - F(x_{\epsilon},\epsilon(x_{\epsilon}-y_{\epsilon}),A_{\epsilon})
                                        \nonumber\\
    &\leq&    \omega(\epsilon\left| y_{\epsilon} - x_{\epsilon} \right|^2 + \left| y_{\epsilon} - x_{\epsilon} \right|)
\end{eqnarray}
for a function $\omega(\zeta) = C \zeta$, with $C = \max\left[C_f,C_g\right]$. The function $\omega: [0,\infty) \to [0,\infty)$, which satisfies the condition $\omega(0^+)=0$, is called a modulus of continuity.
\\

\textbf{Step 6: The Jump Term}\\
We now consider the jump term
\begin{eqnarray}\label{eq_JDRSAM_theo_comp_ineq1_unbd}
   &&
        \frac{1}{\theta}\int_{\mathbf{Z}} \left\{
                e^{-\theta \left(v(t_{\epsilon},y_{\epsilon}+\xi(z)) - v(t_{\epsilon},y_{\epsilon}) \right)}
            -   e^{-\theta \left(u(t_{\epsilon},x_{\epsilon}+\xi(z)) - u(t_{\epsilon},x_{\epsilon}) \right)}
        \right\}\nu(dz)
                                    \nonumber\\
    &=&
        \frac{1}{\theta}\int_{\mathbf{Z}} \left\{
                e^{-\theta \left(v(t_{\epsilon},y_{\epsilon}+\xi(z)) - v(t_{\epsilon},y_{\epsilon}) \right)}
            -   e^{-\theta \left(
                      u(t_{\epsilon},x_{\epsilon}+\xi(z)) - u(t_{\epsilon},x_{\epsilon})
                    + v(t_{\epsilon},x_{\delta}) -v(t_{\epsilon},x_{\delta})
            \right)}
        \right\}\nu(dz)
                                    \nonumber\\
\end{eqnarray}

Since for $\epsilon > 0$ large enough, $u(t,x) - v(t,y) \geq 0$ then
\begin{eqnarray}
        u(t_{\epsilon},x_{\epsilon}+\xi(z))
    -   u(t_{\epsilon},x_{\epsilon})
    +   v(t_{\epsilon},y_{\epsilon})
    -   v(t_{\epsilon},y_{\epsilon}+\xi(z))
    \leq
    -   ( u(t_{\epsilon},x_{\epsilon}) - v(t_{\epsilon},y_{\epsilon}) )
    +   N
                                                \nonumber
\end{eqnarray}
by definition of $N$. Moreover, since $ N_{\epsilon} = \sup_{(t,x,y) \in Q_d}\left[ u(t,x) - v(t,y) - \varphi(t) - \epsilon \lvert x - y \rvert^2 \right]>0$, then $N_{\epsilon} \leq u(t_{\epsilon},x_{\epsilon}) - v(t_{\epsilon},y_{\epsilon})$ and therefore
\begin{equation}
        u(t_{\epsilon},x_{\epsilon}+\xi(z))
    -   u(t_{\epsilon},x_{\epsilon})
    +   v(t_{\epsilon},y_{\epsilon})
    -   v(t_{\epsilon},y_{\epsilon}+\xi(z))
    \leq
        N
    -   N_{\epsilon}
                                                \nonumber
\end{equation}
for $z \in \mathbf{Z}$. Thus,
\begin{eqnarray}
    e^{-\theta \left(
                u(t_{\epsilon},x_{\epsilon}+\xi(z))
            -   u(t_{\epsilon},x_{\epsilon})
            +   v(t_{\epsilon},y_{\epsilon})
            -   v(t_{\epsilon},y_{\epsilon})
        \right)}
    \geq
    e^{-\theta \left(
                v(t_{\epsilon},y_{\epsilon}+\xi(z))
            -   v(t_{\epsilon},y_{\epsilon})
            +   N
            -   N_{\epsilon}
        \right)}
                                                \nonumber
\end{eqnarray}
and equation~\eqref{eq_JDRSAM_theo_comp_ineq1_unbd} can be bounded from above by:
\begin{eqnarray}\label{eq_JDRSAM_theo_comp_ineq2_unbd}
%   &&
%        - \varphi'(t_{\delta})
%                                    \nonumber\\
%   &\leq&
&&
        \frac{1}{\theta}\int_{\mathbf{Z}} \left\{
                e^{-\theta \left(v(t_{\epsilon},y_{\epsilon}+\xi(z)) - v(t_{\epsilon},y_{\epsilon}) \right)}
            -   e^{-\theta \left(u(t_{\epsilon},x_{\epsilon}+\xi(z)) - u(t_{\epsilon},x_{\epsilon})
                    + v(t_{\epsilon},x_{\epsilon}) - v(t_{\epsilon},x_{\epsilon}) \right)}
        \right\}\nu(dz)
                                    \nonumber\\
    &\leq&
        \frac{1}{\theta}\int_{\mathbf{Z}} \left\{
                e^{-\theta \left(v_{\epsilon}(t_{\epsilon},y_{\epsilon}+\xi(z)) - v(t_{\epsilon},y_{\epsilon}) \right)}
            -   e^{-\theta \left(
                v(t_{\epsilon},y_{\epsilon}+\xi(z))
            -   v(t_{\epsilon},y_{\epsilon})
            +   N
            -   N_{\delta}
            \right)}
        \right\}\nu(dz)
                                    \nonumber\\
    &=&
        \frac{1}{\theta}\int_{\mathbf{Z}} \left\{
                e^{-\theta \left(v(t_{\epsilon},y_{\epsilon}+\xi(z)) - v(t_{\epsilon},y_{\epsilon}) \right)}
                \left( 1 - e^{-\theta(N - N_{\epsilon})} \right)
        \right\}\nu(dz)
                                    \nonumber\\
    &=&
       \frac{1}{\theta}\int_{\mathbf{Z}} \left\{
                e^{-\theta \left(-\frac{1}{\theta}\left[\ln\tilde{v}(t_{\epsilon},y_{\epsilon}+\xi(z)) - \ln\tilde{v}(t_{\epsilon},y_{\epsilon})\right] \right)}
                \left( 1 - e^{-\theta(N - N_{\epsilon})} \right)
        \right\}\nu(dz)
                                    \nonumber\\
    &=&
       \frac{1}{\theta}\int_{\mathbf{Z}} \left\{
                \frac{\tilde{v}(t_{\epsilon},y_{\epsilon}+\xi(z))}{\tilde{v}(t_{\epsilon},y_{\epsilon})}
                \left( 1 - e^{-\theta(N - N_{\epsilon})} \right)
        \right\}\nu(dz)
\end{eqnarray}
\\

By Proposition~\ref{prop_JDRSAM_tildePhi_bounded} and since $\tilde{v}$ is LSC, then $\exists \lambda >0 : 0 < \lambda \leq \tilde{v}(t,x) \leq C_{\tilde{\Phi}} \forall (t,x) \in Q$. As a result,
\begin{equation}
    \frac{\tilde{v}(t_{\epsilon},y_{\epsilon}+\xi(z))}{\tilde{v}(t_{\epsilon},y_{\epsilon})} \leq K
                                    \nonumber
\end{equation}
for some constant $K>0$. In addition, since the measure $\nu$ is assumed to be finite and the function $\zeta \mapsto e^{\zeta}$ is continuous, we can establish the following upper bound for the right-hand side of~\eqref{eq_JDRSAM_theo_comp_ineq2_unbd}:
\begin{eqnarray}\label{eq_JDRSAM_theo_comp_ineq3_unbd}
   &&
       \frac{1}{\theta}\int_{\mathbf{Z}} \left\{                \frac{\tilde{v}(t_{\epsilon},y_{\epsilon}+\xi(z))}{\tilde{v}(t_{\epsilon},y_{\epsilon})}
                \left( 1 - e^{-\theta(N - N_{\epsilon})} \right)
        \right\}\nu(dz)
                                                \nonumber\\
   &\leq&
       \frac{K}{\theta}\int_{\mathbf{Z}} 
			\left\{ 1 - e^{-\theta(N - N_{\epsilon})}
        \right\}\nu(dz)
                                                \nonumber\\
   &\leq&
        \omega_R(N-N_{\epsilon})
        \sup_{(t,y) \in [0,T]\times\mathbb{R}^n}
        \nu(\mathbf{Z})
\end{eqnarray}
for some modulus of continuity $\omega_R$ related to the function $\zeta \mapsto 1 - e^{\zeta}$ and parameterized by the radius $R>0$ of the Ball $\mathscr{B}_R$ introduced in Step 3. Note that this parametrization is implicitly due to the dependence of $N$ and $N_{\epsilon}$ on $R$. The term $\sup_{(t,y) \in [0,T]\times\mathbb{R}^n} \nu(\mathbf{Z})$ is the upper bound for the measure $\nu$.
\\

\textbf{Step 7: Conclusion}\\

We now substitute the upper bound obtained in inequalities~\eqref{eq_JDRSAM_theo_comp_modulus_continuity} and~\eqref{eq_JDRSAM_theo_comp_ineq3_unbd} in~\eqref{eq_JDRSAM_theo_comp_bd_ineq1_unbd} to obtain:
\begin{eqnarray}\label{eq_JDRSAM_theo_comp_ineq4_unbd}
    - \varphi'(t_{\epsilon})
    &\leq&  \omega(\epsilon\left| y_{\epsilon} - x_{\epsilon} \right|^2 + \left| y_{\epsilon} - x_{\epsilon} \right|)
            +\omega_R(N-N_{\epsilon})\sup_{(t,x) \in[0,T]\times\mathbb{R}^n} \nu(\mathbf{Z})
\end{eqnarray}
\\

Taking the limit superior in inequality~\eqref{eq_JDRSAM_theo_comp_ineq4_unbd} as $\epsilon \to \infty$ and recalling that
\begin{enumerate}[(1).]
\item the measure $\nu$ is finite;
\item $\xi_i(z), i=1,\ldots, m$ is bounded $\forall z \in \mathbf{Z} \textrm{ a.s. } d\nu$
\end{enumerate}

we see that
\begin{equation}
    \nu(\mathbf{Z})
    < \infty
                                            \nonumber
\end{equation}

Then
\begin{eqnarray}
    \lim_{\epsilon \to 0}\omega_R(N-N_{\epsilon})\nu(\mathbf{Z})
    = 0
                                            \nonumber
\end{eqnarray}
which leads to the contradiction
\begin{equation}
    -\varphi'(t) = \frac{\eta}{t^2} \leq 0
                                                                    \nonumber
\end{equation}
\\

We conclude from this that Assumption~\ref{as_JDRSAM_comparison_theo_contradiction_Phi} is false and therefore
\begin{eqnarray}
    \sup_{(t,x) \in Q} \left[ v(t,x) - u(t,x) \right] \geq 0
\end{eqnarray}
Stated differently, we conclude that
\begin{equation}
    u \leq v
    \quad \textrm{on } [0,T] \times \mathbb{R}^n
                                \nonumber
\end{equation}

\end{proof}

%%%%%%%%%%%%%%%%%%%%%%%%%%%%%%%%%%%%%%%%%%%%%%%%%%%%%%%%%%%%%%%%%%%%%%%%%%%%%%
%   END OF PROOF
%%%%%%%%%%%%%%%%%%%%%%%%%%%%%%%%%%%%%%%%%%%%%%%%%%%%%%%%%%%%%%%%%%%%%%%%%%%%%%

%%%
\subsection{Uniqueness}

Uniqueness is a direct consequence of Theorem~\ref{theo_JDRSAM_comparison_unbounded}. Another important corollary is the fact that the (discontinuous) locally bounded viscosity solution $\Phi$ is in fact continuous on $[0,T] \times \mathbb{R}^n$.

\begin{corollary}[Uniqueness and Continuity]\label{coro_JDRSAM_uniqueness_continuity}
    The function $\Phi(t,x)$ defined on $[0,T] \times \mathbb{R}^n$ is the unique continuous viscosity solution of the RS HJB PIDE~\eqref{eq_JDRSAM_HJBPDE} subject to terminal condition~\eqref{eq_JDRSAM_HJBPDE_termcond}.
\end{corollary}

\begin{proof}
    Uniqueness is a standard by-product of Theorem~\ref{theo_JDRSAM_comparison_unbounded}. Continuity can be proved as follows. By definition of the upper and lower semicontinuous envelopes, recall that

\begin{equation}
    \Phi_*  \leq    \Phi    \leq    \Phi^*
                                            \nonumber
\end{equation}

By Corollary~\ref{coro_JDRSAM_viscositysol} $\Phi_*$ and $\Phi^*$ are respectively semicontinuous superolution and subsolution of the RS HJB PIDE~\eqref{eq_JDRSAM_HJBPDE} subject to terminal condition~\eqref{eq_JDRSAM_HJBPDE_termcond}
\\

We note that as a consequence of Theorem~\ref{theo_JDRSAM_comparison_unbounded} is that
\begin{equation}
    \Phi_*  \geq    \Phi^*
                                            \nonumber
\end{equation}

and hence
\begin{equation}
    \Phi_*  =   \Phi^*
                                            \nonumber
\end{equation}
is a continuous viscosity solution of the RS HJB PIDE~\eqref{eq_JDRSAM_HJBPDE} subject to terminal condition~\eqref{eq_JDRSAM_HJBPDE_termcond}.
\\

Hence, $\Phi = \Phi_* = \Phi^*$ and it is the unique continuous viscosity solution of the RS HJB PIDE~\eqref{eq_JDRSAM_HJBPDE} subject to terminal condition~\eqref{eq_JDRSAM_HJBPDE_termcond}.

\end{proof}

Now that we have proved uniqueness and continuity of the viscosity solution $\Phi$ to the RS HJB PIDE~\eqref{eq_JDRSAM_HJBPDE} subject to terminal condition~\eqref{eq_JDRSAM_HJBPDE_termcond}, we can deduce that the RS HJB PIDE~\eqref{eq_JDRSAM_exptrans_HJBPDE} subject to terminal condition~\eqref{eq_JDRSAM_exptrans_HJBPDE_termcond} also has a unique continuous viscosity solution. We formalize the uniqueness and continuity of $\tilde{\Phi}$ in the following corollary:

%%%%%%%%%%%%%%%%%%%%%%%%%%%%%%%%%%%%%%%%%%%%%%%%%%%%%%%%%%%%%%%%%%%%%%%%%%%%%%
%   COROLLARY
%%%%%%%%%%%%%%%%%%%%%%%%%%%%%%%%%%%%%%%%%%%%%%%%%%%%%%%%%%%%%%%%%%%%%%%%%%%%%%

\begin{corollary}[Uniqueness and Continuity]\label{coro_JDRSAM_uniqueness_continuity_Phi_tilde}
    The function $\tilde{\Phi}(t,x)$ defined on $[0,T] \times \mathbb{R}^n$ is the unique continuous viscosity solution of the RS HJB PIDE~\eqref{eq_JDRSAM_exptrans_HJBPDE} subject to terminal condition~\eqref{eq_JDRSAM_exptrans_HJBPDE_termcond}.
\end{corollary}

%%%%%%%%%%%%%%%%%%%%%%%%%%%%%%%%%%%%%%%%%%%%%%%%%%%%%%%%%%%%%%%%%%%%%%%%%%%%%%%%%%%
%
%   NEXT SECTION
%
%%%%%%%%%%%%%%%%%%%%%%%%%%%%%%%%%%%%%%%%%%%%%%%%%%%%%%%%%%%%%%%%%%%%%%%%%%%%%%%%%%%

%%%
\section{Conclusion}
In this chapter, we considered a risk-sensitive asset management model with assets and factors modelled using affine jump-diffusion processes. This apparently simple setting conceals a number of difficulties, such as the unboundedness of the instantaneous reward function $g$ and the high nonlinearity of the HJB PIDE, which make the existence of classical $C^{1,2}$ solution unlikely barring the introduction of significant assumptions. As a result, we considered a wider class of weak solutions, namely viscosity solutions. We proved that the value function of a class of risk sensitive control problems and established uniqueness by proving a non-standard comparison result. The viscosity approach has proved remarkably useful at solving difficult control problems for which the classical approach may fail. However, it is limited by the fact that it only provides continuity of the value function and by its focus on the PDE in relative isolation from the actual optimization problem. The question is where to go from there? A possible avenue of research would be to look for a method to establish smootheness of the value function, for example through a connection between viscosity solutions and classical solutions. Achieving this objective may also require changes to the analytic setting in order to remove some of the difficulties inherent in manipulating unbounded functions.
\\

%%%%%%%%%%%%%%%%%%%%%%%%%%%%%%%%%%%%%%%%%%%%%%%%%%%%%%%%%%%%%%%%%%%%%%%%%%%%%%%%
%
% BIBLIOGRAPHY
%
%%%%%%%%%%%%%%%%%%%%%%%%%%%%%%%%%%%%%%%%%%%%%%%%%%%%%%%%%%%%%%%%%%%%%%%%%%%%%%%%

\end{document}